\documentclass[a4paper,english,final]{lipics-v2019}
\hypersetup{colorlinks = true,
linkcolor = blue,
urlcolor  = blue,
citecolor = blue,
anchorcolor = blue}

\usepackage{booktabs}   %
\usepackage{subcaption} %

\usepackage{times}
\usepackage{amsmath, amssymb}
\usepackage{pgf}
\usepackage{url}
\usepackage{thmtools}
\usepackage{thm-restate}

\usepackage{bbold}

\usepackage{bm}
\usepackage{mathtools}
\newcommand{\myMatrix}[1]{\bm{\mathit{#1}}}

\usepackage{xargs}                      %

\usepackage[colorinlistoftodos,disable, prependcaption,textsize=tiny]{todonotes}

\pgfkeys{
	/cnfxor/.is family, /cnfxor,
	default/.style = {k = {k}, n = {n}, r = {r}, s = {s}, cnfNum = {}, xorNum = {}},
	k/.estore in = \clauseK,
	n/.estore in = \clauseN,
	r/.estore in = \clauseR,
	s/.estore in = \clauseS,
	cnfNum/.estore in = \clauseCNFOverride,
	xorNum/.estore in = \clauseXOROverride,
}

\newcommand\F[1][]{%
	\pgfkeys{/cnfxor, default, #1}%
	F_{\clauseK}(\clauseN, 
	\ifx\clauseCNFOverride\empty\relax
	\clauseR \clauseN
	\else
	\clauseCNFOverride
	\fi
	)
}

\newcommand\Q[1][]{%
	\pgfkeys{/cnfxor, default, #1}%
	Q(\clauseN, 
	\ifx\clauseXOROverride\empty\relax
	\clauseS \clauseN
	\else
	\clauseXOROverride
	\fi
	)
}

\newcommand\FQ[1][]{%
	\pgfkeys{/cnfxor, default, #1}%
	\phi_{\clauseK}(\clauseN, 
	\ifx\clauseCNFOverride\empty\relax
	\clauseR \clauseN
	\else
	\clauseCNFOverride
	\fi,
	\ifx\clauseXOROverride\empty\relax
	\clauseS \clauseN
	\else
	\clauseXOROverride
	\fi
	)
}

\newcommand{\SAT}{\ensuremath{\mathsf{SAT}}}
\newcommand{\DSharp}{\ensuremath{\mathsf{DSharp}}}

\usepackage{amssymb}
\usepackage{longtable}
\usepackage{amsmath}
\usepackage{graphicx} 
\usepackage{graphics, color}
\usepackage{algorithm}
\usepackage{algpseudocode}
\usepackage{longtable}
\usepackage{tabularx}
\usepackage{url}
\newfloat{algorithm}{t}{lop}
\usepackage{enumitem}
\usepackage{xcolor}

\newcommand{\prob}{\ensuremath{\mathsf{Pr}}}
\newcommand{\expect}{\ensuremath{\mathsf{E}}}

\newcommand{\Cell}[2]{\ensuremath{\mathsf{Cell}_{\langle #1, #2 \rangle}}}
\newcommand{\FullCell}[1]{\ensuremath{\mathsf{Cell}_{\langle #1 \rangle}}}
\newcommand{\Cnt}[2]{\ensuremath{\mathsf{Cnt}_{\langle #1, #2 \rangle}}}
\newcommand{\satisfying}[1]{\ensuremath{sol({#1})}} %
\newcommand{\SatisfyingHashSet}[3]{\ensuremath{\mathsf{Cell}_{\langle #1, #2, #3 \rangle}}}
\newcommand{\SatisfyingHashed}[3]{\ensuremath{|\mathsf{Cell}_{\langle #1, #2, #3 \rangle}|}}
\newcommand{\Hrennes}[1]{\ensuremath{\mathcal{H}^{#1}_{\mathit{Rennes}}}}

\newcommand{\BoundedSAT}{\ensuremath{\mathsf{BoundedCount}}}

\newcommand{\PAC}{\ensuremath{\mathsf{PAC}}}

\newcommand{\NP}{\ensuremath{\mathsf{NP}}}

\newcommand{\ApproxCount}{\ensuremath{\mathsf{ApproxCount}}}

\newcommand{\sharpSAT}{\ensuremath{\#\mathsf{SAT}}}
\newcommand{\sharpP}{\ensuremath{\#\mathsf{P}}}

\newcommand{\ScalApproxMC}{\ensuremath{\mathsf{ApproxMC4}}}
\newcommand{\SparseScalMC}{\ensuremath{\mathsf{ApproxMC5}}}

\newcommand{\SparseScalMCCore}{\ensuremath{\mathsf{ApproxMC5Core}}}
\newcommand{\solCount}{\ensuremath{\mathsf{nSols}}}
\newcommand{\cellCount}{\ensuremath{\mathsf{nCells}}}
\newcommand{\emptyList}{\ensuremath{\mathsf{emptyList}}}

\newcommand{\ApproxMC}{\ensuremath{\mathsf{ApproxMC}}}

\newcommand{\FibBinSearch}{\ensuremath{\mathsf{LogSATSearch}}}

\newcommand{\iter}{\ensuremath{\mathsf{iter}}}
\newcommand{\AddToList}{\ensuremath{\mathsf{AddToList}}}
\newcommand{\FindMedian}{\ensuremath{\mathsf{FindMedian}}}

\newcommand{\iniThresh}{\ensuremath{\mathsf{iniThresh}}}
\newcommand{\hiThresh}{\ensuremath{\mathsf{thresh}}}
\newcommand{\qs}{\ensuremath{\mathsf{qs}}}

\newcommand{\Vars}{\ensuremath{\mathsf{Vars }}}

\newcommand{\thresh}{\ensuremath{\mathsf{thresh}}}

\newcommand{\pivot}{\ensuremath{\mathrm{pivot}}}

\newcommand{\ApproxMCTwo}{\ensuremath{\mathsf{ApproxMC2}}}
\newcommand{\ApproxMCTwoCore}{\ensuremath{\mathsf{ApproxMC2Core}}}
\newcommand{\ApproxMCThree}{\ensuremath{\mathsf{ApproxMC3}}}

\algnotext{EndFor}
\algnotext{EndIf}
\algnotext{EndWhile}
\algblockdefx{Do}{EndDo}{\textbf{do}\;}{}
\algnotext{EndDo}
\usepackage{url}

\newcommand{\eat}[1]{}
\title{Sparse Hashing for Scalable Approximate Model Counting: Theory and Practice
	\protect\thanks{The authors decided to forgo the old convention of alphabetical ordering of authors in favor of a randomized ordering (denoted by \textcircled{r}) . The publicly verifiable record of the randomization is available at \protect\url{https://www.aeaweb.org/journals/policies/random-author-order/search} with confirmation code: lnQDuHqqJDdc. For citation of the work, authors request that the citation guidelines by AEA (available at \url{https://www.aeaweb.org/journals/policies/random-author-order}) for random author ordering be followed. 
	}
}

\titlerunning{Sparse Hashing for Scalable Approximate Model Counting}

\author{Kuldeep S. Meel \protect\footnotemark{} \textcircled{r}  S. Akshay \protect \footnotemark{} }{\addtocounter{footnote}{-2} \footnotemark{} School of Computing, National University of Singapore  \\ \footnotemark{}  Dept of CSE, Indian Institute of Technology, Bombay  }{}{}{}

\authorrunning{K.\,S. Meel \textcircled{r} S. Akshay}

\Copyright{K.\,S. Meel and S. Akshay}



\relatedversion{This manuscript is the full version of paper accepted at LICS2020}%

\supplement{Experimental results and benchmarks are available at https://doi.org/10.5281/zenodo.3766168}%

\funding{Supported in part by National Research Foundation Singapore under its NRF Fellowship Program [NRF-NRFFAI1-2019-0004], NUS ODPRT Grant [R-252-000-685-13], and Sung Kah Kay Assistant Professorship Endowment, DST-SERB MATRICES grant and CEFIPRA Indo-French project EQuaVe.}

\acknowledgements{Part of the work was performed during both authors' stays in Rennes and IIT Delhi, and Meel's visits to IIT Bombay; The authors  are grateful to INRIA Rennes, IIT Delhi, and IIT Bombay for the hospitality. Meel owes gratitude to Ashish Gupta for organizing a memorable trip and the ensuing discussions that provided crucial insights into linking variance to the structure  of the solution space. The authors would also like to thank Yash Pote, Mate Soos, and Bhavishya for their generous help in experimental evaluation,  Supratik Chakraborty and Aditya Shrotri for providing crucial feedback and in particular, discovery of a bug; the repair of which was possible due to generous help from Cyrus Rashtchian and Paul Beame.}%

\hideLIPIcs  %

\newcounter{casenum}
\newenvironment{caseof}{\setcounter{casenum}{1}}{\vskip.5\baselineskip}
\newcommand{\case}[2]{\vskip.5\baselineskip\par\noindent {\bfseries Case \arabic{casenum}:} #1\\#2\addtocounter{casenum}{1}\begin{flushright}
	
\end{flushright}
}

\begin{document}
\maketitle
\begin{abstract}
  Given a CNF formula $F$ on $n$ variables, the problem of model counting, also referred to as $\#SAT$, is to compute the number of models or satisfying assignments of $F$. Model counting is a fundamental but hard problem in computer science with varied applications.
Recent years have witnessed a surge of effort towards developing efficient algorithmic techniques that combine the classical 2-universal hashing (from~\cite{Stockmeyer83}) with the remarkable progress in SAT solving over the past decade.  These techniques augment the CNF formula $F$ with random XOR constraints and invoke an NP oracle repeatedly on the resultant CNF-XOR formulas. In practice, the NP oracle calls are replaced by calls to a SAT solver and it is observed that runtime performance of modern SAT solvers (based on conflict-driven clause learning) on CNF-XOR formulas is adversely affected by the size of XOR constraints. The standard construction of 2-universal hash functions chooses every variable with probability $p=\frac{1}{2}$ leading to XOR constraints of size $\frac{n}{2}$ in expectation. Consequently, the main challenge is to design \emph{sparse} hash functions, where variables can be chosen with smaller probability and lead to smaller sized XOR constraints, which can then replace 2-universal hash functions. 
  
In this paper, our goal is to address this challenge both from a theoretical and a practical perspective. First, we formalize a relaxation of universal hashing, called concentrated hashing, a notion implicit in prior works to design sparse hash functions. We then establish a novel and beautiful connection between concentration measures of these hash functions and isoperimetric inequalities on boolean hypercubes. This allows us to obtain tight bounds on variance as well as the dispersion index and show that $p = \mathcal{O}(\frac{\log_2 m}{m})$ suffices for the design of sparse hash functions from $\{0,1\}^n$ to $\{0,1\}^m$ belonging to the concentrated hash family. Finally, we use sparse hash functions belonging to this concentrated hash family to develop new approximate counting algorithms. A comprehensive experimental evaluation of our algorithm on 1893 benchmarks demonstrates that the usage of sparse hash functions can lead to significant speedups.  To the best of our knowledge, this work is the first study to demonstrate runtime improvement of approximate model counting algorithms through the usage of sparse hash functions, while still retaining strong theoretical guarantees (\`a la 2-universal hash functions).
\end{abstract}

\section{Introduction}\label{sec:intro}
Given a Boolean formula $F$ in conjunctive normal form (CNF), the problem of model counting, also referred to as {\sharpSAT}, is to compute the number of models of $F$. Model counting is a fundamental problem in computer science with a wide variety of applications ranging from quantified information leakage~\cite{FJ14}, probabilistic reasoning~\cite{Roth1996,SangBeamKautz2005,CFMSV14,EGSS13c}, network reliability~\cite{Valiant79,DMPV17}, neural network verification~\cite{BSSMS19}, and the like.  For example, given a probabilistic model describing conditional dependencies between different variables in a system, the problem of probabilistic inference, which seeks to compute the probability of an event of interest given observed evidence, can be reduced to a collection of model counting queries~\cite{Roth1996}. 

In his seminal paper, Valiant showed that {\sharpSAT} is $\sharpP$-complete, where {\sharpP} is the set of counting problems associated with {\NP} decision problems ~\cite{Valiant79}. Theoretical investigations of {\sharpP} have led to the discovery of deep connections in complexity theory, and there is strong evidence for its hardness \cite{AroBar09,Toda89}. In particular, Toda showed that every problem in the polynomial hierarchy could be solved by just one call to a {\sharpP} oracle; more formally, $PH \subseteq P^{\sharpP}$~\cite{Toda89}. 

Given the computational intractability of {\sharpSAT}, researchers have focused on approximate variants. Stockmeyer presented a randomized hashing-based technique that can compute $(\varepsilon,\delta)$ approximation within the polynomial time, in $|F|, \varepsilon, \delta$,  given access to a {\NP} oracle where $|F|$ is the size of formula, $\varepsilon$ is the error tolerance bound and $\delta$ is the confidence\footnote{Although Stockmeyer did not present a randomized variant in his 1983 paper, Jerrum, Valiant, and Vazirani credit Stockmeyer for the idea~\cite{JVV86}}. The computational intractability of {\NP} dissuaded development of algorithmic implementations of Stockmeyer's hashing-based techniques and no practical tools for approximate counting existed until the 2000's~\cite{GSS06}. By extending Stockmeyer's framework, Chakraborty, Meel, and Vardi demonstrate a scalable $(\varepsilon,\delta)$-counting algorithm, {\ApproxMC} \cite{CMV13b}. Subsequently, several new algorithmic ideas have been incorporated to demonstrate the scalability of {\ApproxMC}; the current version of {\ApproxMC} is called {\ScalApproxMC}~\cite{CMV16,SM19,SGM20}. Recent years have seen a surge of interest in the design of hashing-based techniques for approximate counting~\cite{EGSS13a,EGSS13c,CFMSV14,IMMV15,MVCFSFIM16,CMMV16,SM19,SGM20}. 

The core theoretical idea of the hashing-based framework is to employ 2-universal hash functions to partition the solution space, denoted by $\satisfying{F}$ for a formula $F$,  into {\em roughly equal small} cells, wherein a cell is called {\em small} if it has solutions less than or equal to a  pre-computed threshold, $\thresh$. An {\NP} oracle is employed to check if a cell is small by enumerating solutions one-by-one until either there are no more solutions or we have already enumerated $\thresh+1$ solutions. To ensure polynomially many NP calls, $\thresh$ is set to be  polynomial in input parameter $\varepsilon$. The choice of the threshold gives rise to a tradeoff between the number of \NP\ queries and size of each query. To achieve probabilistic amplification of the confidence, multiple invocations of underlying subroutines are performed. 

A standard family of 2-universal hash functions employed for this is the $H_{xor}$ family comprising of functions expressed as conjunction of XOR constraints.  In particular, viewing the set of variables $Y$ of the formula $F$ as a vector of dimension $n\times 1$, one can represent the hash function $h:\{0,1\}^n \mapsto \{0,1\}^m$ as $h(Y) = \myMatrix{A}Y+\myMatrix{b}$ where $\myMatrix{A}$ is a $m\times n$ matrix while $\myMatrix{b}$ is $m\times1$ 0-1 vector and each entry of $\myMatrix{A}$ and $\myMatrix{b}$ is either 0 or 1.  Each entry of $\myMatrix{A}$ is chosen to be 1 with probability $p= 1/2$, therefore the average number of 1's in each row is $\frac{n}{2}$.  Each row of $h(Y)$ thus gives rise to XOR constraints involving $\frac{n}{2}$ variables in expectation.
Similarly a cell $\alpha$ can be viewed as a 0-1 vector of size $m\times 1$. 
Now, the solutions of $F$ in a given cell $\alpha$ are the solutions of the formula $F \wedge (\myMatrix{A}Y+\myMatrix{b} = \alpha)$. As the input formula $F$ is in CNF, this formula is a conjunction of CNF and XOR-constraints, also called an CNF-XOR formula. Given a hash function $h$ and a cell $\alpha$, the random variable of interest, denoted by $|\Cell{F}{h,\alpha}|$ is the number of solutions of $F$ that $h$ maps to cell $\alpha$. As mentioned earlier, the NP-oracle is invoked (polynomially many times) to check if such a cell is small.

The practical implementation of these techniques employ a {\SAT} solver to perform {\NP} oracle calls. The performance of SAT solvers, however, degrades with increase in the number of variables in XOR constraints (also called their \emph{width}) and therefore recent efforts have focused on design of \emph{sparse} hash functions where each entry is chosen with $p \ll 1/2$ ($p$ is also referred to as density)~\cite{GHSS07,EGSS13c,IMMV15,AD16,AHT18,AT17}. The primary theoretical challenge is that 2-universality has been crucial to obtain $(\varepsilon,\delta)$-guarantees, and sparse hash functions are not 2-universal. In fact, despite intense theoretical and practical interest in the design of sparse hash functions, the practical implementation of all prior constructions have had to sacrifice theoretical guarantees (as further discussed in Section~\ref{sec:stateofart}). 

Given the applications of counting to critical domains such as network reliability, the loss of theoretical guarantees limits the applications of approximate model counters. Therefore, in this context, the main challenge is: {\bf Is it possible to construct sparse hash functions and design algorithmic frameworks to achieve runtime performance improvement without losing theoretical guarantees?}

In this paper, we address this challenge. To this end, we formalize the implicit observation in prior works that hashing-based counting algorithms, similar to other applications of universal hashing, are primarily concerned with the application of concentration bounds. We start by providing, in Section~\ref{sec:prelims}, a definition of concentrated hash functions, a relaxation of universal hashing. The guarantees offered by concentrated hashing depend crucially on the size of the set, unlike in universal hashing. Next, we turn towards the construction of sparse hash functions that belong to the concentrated hash family. Finally, we explain how these sparse hash functions can be used to build an efficient algorithm for approximate model counting. More precisely, the technical contributions of this paper are the following:
\begin{enumerate}
\item 
We first obtain a characterization of $\satisfying{F}$ that would achieve the maximum variance as well as dispersion index for $|\Cell{F}{h, \alpha}|$ for sparse hash functions. In a significant departure from earlier works~\cite{EGSS14a,AD16,ZCSE16,AHT18} where the focus was to use analytical methods to obtain upper bound on the variance of $|\Cell{F}{ h, \alpha}|$, we focus on searching for the set $\satisfying{F}$ that would achieve the maximum variance of $|\Cell{F}{h, \alpha}|$.  To do this, we utilize a beautiful connection between the maximizing of variance as well as dispersion index of $|\Cell{F}{h,\alpha}|$ and minimizing the ``$t$-boundary'' (the number of pairs with Hamming distance at most $t$) of sets on the boolean hypercube on $n$ dimensions. This allows us to obtain novel and stronger upper bounds by using deep results from Boolean functional analysis and isoperimetric inequalities~\cite{BR17,R18}. This connection could possibly be applied in other contexts as well. 
\item Utilizing  the connection between dispersion index and ``$t$-boundary'' allows us to introduce a new family of hash functions, denoted by $\Hrennes{}$, which consists of hash functions of the form $\myMatrix{A}X+\myMatrix{b}$, where every entry of $\myMatrix{A}[i]$ is set to 1  with $p_i = \mathcal{O}(\frac{\log_2 i}{i})$. The construction of the new family marks a significant departure from prior families in the behavior of the density dependent on rows of the matrix $\myMatrix{A}$.  We believe $\Hrennes{}$ is of independent interest and can be substituted for 2-universal hash functions in several applications of hashing.  
\item Finally, we use the above concentrated hash family to develop a new approximate model counting algorithm \SparseScalMC, building on the existing state-of-the-art algorithm {\ScalApproxMC}. The primary challenge lies in the design and analysis of a hashing-based algorithm that does not assume any bound on  $|\satisfying{F}|$ but is able to use concentrated hash functions whose behavior depends on the size of the set being hashed. A comprehensive experimental evaluation on 1893 benchmarks demonstrates  that usage of $\Hrennes{}$   in {\SparseScalMC} leads to significant speedup in runtime over \ScalApproxMC. It is worth viewing the runtime improvement in the context of prior work where significant slowdown was observed. To the best of our knowledge, {\em this work is the first study to demonstrate runtime improvement through sparse hash functions without loss of $(\varepsilon,\delta)-$guarantees, demonstrating the tightness of our bounds in practice. }
\end{enumerate}  

\subsubsection*{Structure of the paper}
We define notations and preliminaries in Section~\ref{sec:prelims} along with a survey of state of the art for design of sparse hash functions in the context of approximate model counting. We then outline the main technical contributions of this paper in Section~\ref{sec:mainresults}. In Section~\ref{sec:concentrated}, we utilize deep results from Boolean functional analysis and isoperimetric inequalities to bound the dispersion index as well as variance of $|\Cell{F,h}{\alpha}|$. We then use the bounds on dispersion index to construct sparse hash families belong to concentrated hashing in Section~\ref{sec:rennes}.  Section~\ref{sec:sparsescalmc} deals with construction of approximate model counting algorithm that uses hash functions belong to concentrated family. We finally describe extensive empirical evaluation in Section~\ref{sec:evaluation} and conclude in Section~\ref{sec:conclusion}. %

\section{Definitions and State of the Art}\label{sec:prelims}

\subsubsection{The model counting problem}
Let $F$ be a Boolean formula in conjunctive normal form (CNF), and let $\Vars(F)$ be the set of variables appearing in $F$.  The set $\Vars(F)$ is also called the \emph{support} of $F$. An assignment $\sigma$ of truth values to the variables in $\Vars(F)$ is called a \emph{satisfying assignment} or \emph{witness} of $F$ if it makes $F$ evaluate to true.  We denote the set of all witnesses of
$F$ by $\satisfying{F}$. Throughout the paper, we will use $n$ to denote $|\Vars(F)|$.

We write $\prob\left[\mathcal{Z}: {\Omega} \right]$ to denote the probability of
outcome $\mathcal{Z}$ when sampling from a probability space ${\Omega}$.  For
brevity, we omit ${\Omega}$ when it is clear from the context.  The
expected value of $\mathcal{Z}$ is denoted $\expect\left[\mathcal{Z}\right]$
and its variance is denoted $\sigma^2\left[\mathcal{Z}\right]$. The quantity $\frac{ \sigma^2\left[\mathcal{Z}\right]}{\expect\left[\mathcal{Z}\right]}$ is called the dispersion index of the random variable $\mathcal{Z}$.  Given a distribution $\mathcal{D}$, we use $\mathcal{Z} \sim \mathcal{D}$ to denote that $\mathcal{Z}$ is sampled from the distribution $\mathcal{D}$. Let Bern(p) denote the Bernoulli distribution with probability $p$ such that if $\mathcal{Z} \sim $Bern(p), we have $\prob[\mathcal{Z} = 1] = p$. 	

The \emph{propositional model counting problem} is to compute
$|\satisfying{F}|$ for a given CNF formula $F$.  A \emph{probably approximately correct} (or \PAC) counter is a probabilistic algorithm ${\ApproxCount}(\cdot, \cdot,\cdot)$ that takes as inputs a formula $F$,  a tolerance $\varepsilon>0$, and a confidence  $\delta\in (0, 1]$, and returns a $(\varepsilon,\delta)$-estimate $c$, i.e.,  $\prob\Big[\frac{|\satisfying{F}|}{1+\varepsilon} \le c \le (1+\varepsilon)|\satisfying{F}|\Big] \ge 1-\delta$. PAC guarantees are also sometimes referred to as $(\varepsilon,\delta)$-guarantees.

A closely related notion is of projected model counting wherein we are interested  in computing the cardinality of $\satisfying{F}$ projected to a subset of variables $\mathcal{P} \subseteq \Vars(F)$. While for clarity of exposition, we focus on the problem of model counting, the techniques developed in this paper apply to projected model counting as well. In our empirical evaluation, we consider such benchmarks as well.

\subsubsection{Universal hash functions}

Let $n,m\in \mathbb{N}$ and $\mathcal{H}(n,m) \triangleq \{ h:\{0,1\}^{n} \rightarrow \{0,1\}^m \}$ be a family of hash functions mapping $\{0,1\}^n$ to $\{0,1\}^m$. We use $h \xleftarrow{R} \mathcal{H}(n,m)$ to denote the probability space obtained by choosing a function $h$ uniformly at random from $\mathcal{H}(n,m)$. To measure the quality of a hash function we are interested in the set of elements of $S$ mapped to $\alpha$ by $h$, denoted $\SatisfyingHashSet{S}{h}{\alpha}$ and its cardinality, i.e., $\SatisfyingHashed{S}{h}{\alpha}$. 

\begin{definition}
  A family of hash functions $\mathcal{H}(n,m)$ is \emph{strongly 2-universal}
      \footnote{The concept of 2-universal hashing proposed by Carter and Wegman~\cite{carter1977universal} only required that $\prob[h(x) =  h(y)] \leq  \frac{1}{2^m}$ and therefore, the phrase {\em strongly 2-universal} is often used as also noted by Vadhan in~\cite{V12}.} if $\forall x,y \in \{0,1\}^n$, $\alpha\in \{0,1\}^m, h \xleftarrow{R} \mathcal{H}(n,m)$, 
	\begin{align}
	\prob[h(x) = \alpha] = \frac{1}{2^m} =	\prob[h(x) =  h(y)] %
	\end{align}
\end{definition}	

\begin{proposition}\label{lm:universal-bound}
	Let $\mathcal{H}(n,m)$ be a strongly 2-universal hash family and let $h \xleftarrow{R} \mathcal{H}(n,m)$, then $\forall S \subseteq \{0,1\}^n$, $|S|\geq 1$, 
	\begin{align}
	\expect[\SatisfyingHashed{S}{h}{\alpha}] = \frac{|S|}{2^m}\\
	\frac{\sigma^2[\SatisfyingHashed{S}{h}{\alpha}]}{ \expect[\SatisfyingHashed{S}{h}{\alpha}]} \leq  1
        \label{eq:dispindex}
	\end{align}
\end{proposition}
 Equation~(\ref{eq:dispindex}) can thus be restated as saying that for universal hash functions, the dispersion index must be at most 1.

\subsubsection{Prefix hash families}
While universal hash families have nice concentration bounds, they are not adaptive, in the sense that one cannot build on previous queries. In several applications of hashing, the dependence between different queries can be exploited to extract improvements in theoretical complexity and runtime performance. Thus, we are typically interested in a restricted class of hash functions, called a \emph{prefix-family} of hash functions defined in~\cite{CMV16} as follows.  For $\alpha \in \{0,1\}^m$, $\alpha[i]$ represent $i$-th element of $\alpha$. 

\begin{definition}\label{def:prefix-family}
  Let $n\in \mathbb{N}$ and $\mathcal{H}(n,1)$ be a family of hash functions. A family of hash functions $\mathcal{H}(n,n)$ is called a \emph{prefix-family} with respect to $\mathcal{H}(n,1)$ if for all $h \in \mathcal{H}(n,n)$, there exists $h_1, h_2, \cdots h_n \in \mathcal{H}(n,1)$ such that 
	\begin{enumerate}
		\item 		$h(x)[i] = h_{i}(x) $
		\item  for all $i \in [n]$, the probability spaces for $\{ h_i \mid h \xleftarrow{R} \mathcal{H}(n,n) \}$ and $\{ g \mid g \xleftarrow{R} \mathcal{H}(n,1)   \}$ are identical. 
	\end{enumerate} 
	
	For every $m \in \{1, \ldots
	n\}$, the $m^{th}$ prefix-slice of $h$, denoted $h^{(m)}$, is a
	map from $\{0,1\}^{n}$ to $\{0,1\}^m$, such that $h^{(m)}(y)[i] = h_{i}(y)$, for all $y \in \{0,1\}^{n}$ and for all $i \in \{1, \ldots
	m\}$. Similarly, the $m^{th}$ prefix-slice of $\alpha$, denoted
	$\alpha^{(m)}$, is an element of $\{0,1\}^m$ such that
	$\alpha^{(m)}[i] = \alpha[i]$ for all $i \in \{1, \ldots m\}$. In this paper we will primarily be focussed on prefix-hash functions and concentration bounds on them. To avoid cumbersome terminology, we abuse notation and write $\Cell{S}{m}$ (resp. $\Cnt{S}{m}$) as a short-hand for $\SatisfyingHashSet{S}{h^{(m)}}{\alpha^{(m)}}$ (resp. $\SatisfyingHashed{S}{h^{(m)}}{\alpha^{(m)}}$).
        \end{definition}
\begin{table}[h]
	\begin{tabular}{|c|c|c|}
		\hline 
		Symbol & Short for & Meaning \\ \hline 
		$\Cell{S}{m}$ & $\SatisfyingHashSet{S}{h^{(m)}}{\alpha^{(m)}}$ &$ S \cap \{  y \mid h^{(m)}(y) = \alpha^{(m)} \}   $\\ \hline 
		$\Cnt{S}{m}$ & $\SatisfyingHashed{S}{h^{(m)}}{\alpha^{(m)}}$ &$ |\Cell{S}{m} | $ \\ \hline 
		
	\end{tabular}
   \caption{List of Important Notations}
\end{table}
In what follows, for a formula $F$, we write $\Cell{F}{m}$ (resp. $\Cnt{F}{m}$)  to mean $\Cell{\satisfying{F}}{m}$ (resp. $\Cnt{\satisfying{F}}{m}$). Finally, the usage of prefix-family ensures monotonicity of the random variable, $\Cnt{S}{i}$, since from the definition of prefix-family, we have that for all $i$, $h^{(i+1)}(x) = \alpha^{(i+1)} \implies h^{(i)}(x) = \alpha^{(i)} $. Formally,
\begin{proposition}\label{prop:prefix-monotonicity}
	 For all $1\leq i <m$, $\Cell{S}{i+1} \subseteq \Cell{S}{i}$
\end{proposition}

\subsubsection{Explicit families and sparse hash functions}
While the above definitions of hash families are abstract, applications to model counting need explicit hash functions. The most common explicit hash family used for this are as follows: Let $\mathcal{H}_{\left\{p_i\right\}_{1\leq i\leq m}} \triangleq \{ h:\{0,1\}^{n} \rightarrow \{0,1\}^m \}$ be the family of functions of the form $h(x) = \myMatrix{A}x+\myMatrix{b}$ with $\myMatrix{A} \in \mathbb{F}_{2}^{m \times n}$ and $\myMatrix{b} \in \mathbb{F}_{2}^{m \times 1}$ where the entries of  $\myMatrix{A}[i]$ and $\myMatrix{b}$ are independently generated according to Bern($p_i$) and Bern($\frac{1}{2}$) respectively. Note that taking $p_i=\frac{1}{2}$ gives $\mathcal{H}_{\{ \frac{1}{2}, \frac{1}{2}, \cdots \frac{1}{2}  \} }(n,m)$, which is precisely the  strongly 2-universal hashing family proposed by Carter and Wegman~\cite{carter1977universal}, also denoted as $H_{xor}(n,m)$ in earlier works~\cite{MVCFSFIM16}. $p_i$ is referred to as the density of $i$-th row of $\myMatrix{A}$ and $1-p_i$ is referred to as the sparsity of $i$-th row of $\myMatrix{A}$. We will use the term {\em sparse hash functions} to refer to hash functions with $p_i \ll \frac{1}{2}$.

Observe that $\mathcal{H}_{\left\{p_i\right\}_{1\leq i\leq n}}$ is a prefix-family with  $h^{(m)}(x) = \myMatrix{A}^{(m)}x+\myMatrix{b}^{(m)}$, where $\myMatrix{A}^{(m)}$ denotes the submatrix formed by the first $m$ rows and $n$ columns of $\myMatrix{A}$ and $\myMatrix{b}^{(m)}$ is the first $m$ entries of the vector $\myMatrix{b}$.

\subsection{Concentrated hash functions} 
 Several applications such as sketching and counting~\cite{Stockmeyer83,CM05} involving universal hash functions invoke strongly 2-universality property solely to obtain Proposition~\ref{lm:universal-bound}, i.e., obtain strong concentration bounds, but as mentioned above this requires fixing $p_i=\frac{1}{2}$.
 
In this context, one might ask if one can relax the requirement of 2-universality, while still attaining similar bounds for expectation and dispersion index. In a spirit similar to other attempts to design sparse hash functions for approximate counting techniques, we seek to design hash functions whose behavior depends on the size of $|S|$. To this end, we formalize the concept of concentrated hash family.

\begin{definition}
Let $\qs,k\in\mathbb{N}$, $\rho\in (0,1/2]$. A family of hash functions $\mathcal{H}(n,n)$ is \emph{prefix-$(\rho,\qs, k)$-concentrated}, if for each $m$ with $\qs \leq m \leq n$,  and $ S \subseteq \{0,1\}^n$ where $|S| \leq k \cdot 2^m$, $\alpha \in \{0,1\}^n, h \xleftarrow{R} \mathcal{H}$, we have 
	\begin{align}
	\expect[\Cnt{S}{m}] = \frac{|S|}{2^m}\\
	\frac{\sigma^2[\Cnt{S}{m}]}{\expect[\Cnt{S}{m} ]} \leq \rho 
	\end{align}
\end{definition}

It is easy to see that this definition is monotonic in $k$ and it generalizes strongly 2-universal hash functions.  Note that the above definition differs from the property of strongly 2-universal hash functions in two ways: first, it  bounds the dispersion index by a constant instead of 1, and second, the definition depends on size of $S$.

\begin{restatable}{proposition}{concmono}
\label{prop:conc-mono} 
		 If $\mathcal{H}(n,n)$ is prefix-$(\rho,\qs, k)$-concentrated, then $\mathcal{H}(n,n)$ is prefix-$(\rho', \qs', k')$-concentrated for all $\rho' \geq \rho$,  $\qs' \geq \qs$, and $k' \leq k$.  
\end{restatable}

Finally, we may show that applying the usual Chebyshev and Paley-Zymund inequalities to this definition immediately gives us the following properties of concentrated hash families.%

\begin{restatable}{proposition}{probbounds}
	\label{prop:probBounds}
	If $\mathcal{H}$ is prefix-$(\rho, \qs, k)$-concentrated family, then for every  $0 < \beta < 1$, $\qs \leq m \leq n$,  and for all $|S| \leq 2^m\cdot k$,  we have the
	following:
	\begin{enumerate}
		\item {\small $\prob\left[\left|\Cnt{S}{m} - \expect[\Cnt{S}{m} ]\right|\geq \beta \expect[\Cnt{S}{m} ]\right] \leq \frac{\rho}{\beta^2 \expect[\Cnt{S}{m} ]}$}
		\item {\small $\prob\left[\Cnt{S}{m} \leq  \beta \expect\left[\Cnt{S}{m} \right]\right]  \leq  \frac{\rho}{\rho+(1-\beta)^2 \expect\left[\Cnt{S}{m} \right]}$}
	\end{enumerate} 
\end{restatable}

Indeed, the rationale behind the design of $(\rho,k)$-concentrated hash families is that one can design such families with significant sparsity. Such sparse hash functions can then contribute to runtime performance of the underlying applications. The notion of concentrated hashing bears some similarity to the notion of {\em  strongly concentrated} random variables defined in ~\cite{EGSS14a}.  In particular, a prefix $(\rho,\qs, k)$ concentrated family implies that the random variable $\Cnt{S}{m}$, for $m \geq \qs$, is  strongly-$\left( (\beta \expect[\Cnt{S}{m}])^2, \frac{\beta^2 \expect[\Cnt{S}{m} ]}{\rho}\right)$ concentrated. We refer the reader to the Appendix A.1 for the formal statement as well as its relations to other useful notions of hashing.

\subsection{State of the Art}\label{sec:stateofart}
The current state of the art hashing-based techniques for approximate model counting can be broadly classified into two categories: the first category of techniques~\cite{trevisan2002lecture,EGSS13b,AT17,AHT18}, henceforth called Cat1,  compute a constant factor approximation by setting $\thresh$ to be a constant and use Stockmeyer's trick of constructing multiple copies of the input formula. The second class of techniques, henceforth called Cat2, consists of techniques~\cite{CMV13b,CMV16,MVCFSFIM16} that directly compute an $(\varepsilon, \delta)$-estimate by setting  $\mathrm{threshold} = \mathcal{O}(\frac{1}{\varepsilon^2})$, and hence invoking the underlying {\NP} oracle $\mathcal{O}(\frac{1}{\varepsilon^2})$ times. The proofs of correctness for all the hashing-based techniques involve the usage of concentration bounds due to strong 2-universal hash functions. Recall that given a hash function $h \in \mathcal{H}(n,m)$ and a cell $\alpha$, the random variable of interest is $\Cnt{F}{m}$ the number of solutions of $F$ that $h$ maps to cell $\alpha$. The Cat1 techniques require the coefficient of variation, defined as the ratio of standard deviation of $\Cnt{F}{m}$ to  $\expect[\Cnt{F}{m}]$, to be upper bounded by a constant  while, for Cat2 techniques, it is sufficient to have the dispersion index be bounded by a constant.  It is worth noting that the analyses for both the techniques allow one to focus on the case of  $\expect[\Cnt{F}{m}]$ being greater than 1. In this case, if dispersion index is upper bounded by a constant, then so is the coefficient of variation (but not vice versa!).  In this sense, Cat2 techniques are stronger than Cat1.

Recently,~\cite{AD16} and~\cite{ZCSE16} independently showed that 2-universality can be relaxed while using Cat1 techniques. More precisely, they showed that choosing entries with probability $p = \mathcal{O}(\log n/n)$ asymptotically suffices to guarantee that the coefficient of variation is upper bounded by constant, i.e., dispersion index is upper bounded by mean of $\Cnt{F}{m}$  when $\log (|\satisfying{F}|) \in \Omega (n)$. Furthermore,~\cite{AT17} showed that (sparse) hash functions constructed using LDPC codes also asymptotically suffice to guarantee that the coefficient of variation is upper bounded by constant. However, these results come with three caveats:

\begin{enumerate}
\item  Only Cat1 techniques can employ these sparse hash functions as they can provide upper bound on coefficient of variation but not dispersion index. On the other hand, Cat2 techniques scale significantly better than Cat1 techniques in practice. ~\cite{BAM20}
\item The asymptotically large constant in the upper bound of coefficient of variation makes the practical usage usage of the above hash functions infeasible as discussed extensively in prior work (cf: Section 9 of~\cite{AHT18}). 
\item The results only hold true for $\log (|\satisfying{F}|) \in \Omega (n)$, which is usually not the case for many practical applications. 
\end{enumerate}
In summary, when $p<\frac{1}{2}$, previous techniques are unable to obtain a constant upper bound on the dispersion index and therefore do not yield to usage in Cat2 techniques (and hence in developing efficient practical algorithms for approximate model counting).

\section{Main Results}\label{sec:mainresults}
To accomplish the design of scalable approximate counters via sparse hashing, we follow a three step recipe: (i) derive an expression to bound the dispersion index (of the random variable $\Cnt{S}{m}$) via boolean functional analysis and isoperimetric inequalities, (ii) construct a sparse $(\rho,k)$-concentrated hash family and (iii) design an approximate model counter which can take advantage of concentrated hashing. In this section, we highlight our strategy, the core ideas involved and the main theorem statements. 

\subsection{Bounding the Dispersion Index}
The first step is to obtain a closed form expression for the upper bound on dispersion index for an arbitrary set $S \subseteq \{0,1\}^n$. To this end, we focus on obtaining an expression that depends on $n$, $|S|$ and the range of hash function, i.e., $m$ for $h^{(m)}$. 

For $1\leq i\leq n-1$, $p_i\in(0,\frac{1}{2}]$, consider the family $\mathcal{H}_{\{p_i\}} (n,n) \triangleq \{ h:\{0,1\}^{n} \rightarrow \{0,1\}^n \}$ of functions of the form $h(x) = \myMatrix{A}x+\myMatrix{b}$ with $\myMatrix{A} \in \mathbb{F}_{2}^{n \times n}$ and $\myMatrix{b} \in \mathbb{F}_{2}^{n \times 1}$ where the entries of $\myMatrix{A}[i]$ (for $1\leq i\leq n$) and $\myMatrix{b}$ are independently generated according to Bern($p_i$) and Bern($\frac{1}{2}$) respectively.	
For $1 \leq m \leq n$, let 
\begin{align*}
q(w,m)  &= \prod_{j=1}^{m}\left(\frac{1}{2}+\frac{1}{2}(1-2p_j)^w\right) 
 \\
r(w,m) &= q(w,m)- \frac{1}{2^m} 
\end{align*}
Note that $r(w,m)$ is a decreasing function of $w$ for a fixed $m$. With this we have the following bound on the dispersion index, which is one of the main technical contributions of this paper, of possible independent interest.

 \begin{restatable}{theorem}{mainexp}
    \label{thm:main-dispexpression}
For $1 \leq m \leq n$, $S \subseteq \{0,1\}^n$,  $\frac{\sigma^2[\Cnt{S}{m}]}{\expect[\Cnt{S}{m}]} \leq 
\sum\limits_{w=0}^{\ell} 2 \cdot \left(\frac{8e \sqrt{n\cdot\ell}}{w}\right)^w r(w,m)$ where $\ell = \lceil \log |S| \rceil$. 
\end{restatable}

 A key ingredient of the proof is to relate the dispersion index (and the variance) of $\Cnt{S}{m}$ to the Hamming distance between nodes of $S$. This allows us to show that the dispersion index is in fact maximized for a nicely behaved set (formally, a left compressed down set as formalized in Section~\ref{sec:concentrated}). Now we invoke deep results from boolean functional analysis and isoperimetric inequalities~\cite{BR17,R18,RR19}, to bound the maximum value of the dispersion index.
 
We remark that the best known bounds for the dispersion index from prior work so far has been: for any $S \subseteq \{0,1\}^n$ , $\frac{\sigma^2[\Cnt{S}{m}]}{\expect[\Cnt{S}{m}]} \leq \sum\limits_{w=0}^{\ell} {n \choose w} q(w,m)$. Since $\left(\frac{8e \sqrt{n\cdot\ell}}{w}\right)^w \leq \cdot {8e \sqrt{n\cdot\ell} \choose w}$, we obtain an improvement from ${n \choose w}$  to $2\cdot {8e \sqrt{n\cdot\ell} \choose w}$.  This improvement combined with our new analysis of the bounds leads us to design sparse hash family without incurring large overhead. It is also worth pointing out that prior work has always upper bounded $r(w,m)$ by $q(w,m)$ but as our analysis in the next section shows, we obtain stronger bounds on the dispersion index due to careful manipulation of $r(w,m)$. 

\subsection{Construction of Sparse Concentrated Hash Family}
The upper bound on dispersion index provided by Theorem~\ref{thm:main-dispexpression} depends on $|S|$, and therefore we turn to the notion of concentrated family for construction of sparse hash functions to capture dependence on $|S|$. To bound the dispersion index, we seek to increase the rate of decrease of the values of $r(w,m)$ with respect to $m$. To this end, we propose a hash family with varying density across different rows of the matrix. 
\begin{definition}
  Let $k,n\in \mathbb{N}$ and let $H^{-1}: [0,1] \rightarrow [0,\frac{1}{2}]$ be the inverse binary entropy function restricting its domain to $[0,\frac{1}{2}]$ so that the inverse is well defined. We then define $\Hrennes{k}(n,n) \triangleq \{ h:\{0,1\}^{n} \rightarrow \{0,1\}^n \}$ to be the family of functions of the form $h(x) = \myMatrix{A}x+\myMatrix{b}$ with $\myMatrix{A} \in \mathbb{F}_{2}^{n \times n}$ and $\myMatrix{b} \in \mathbb{F}_{2}^{n \times 1}$ where the entries of $\myMatrix{A}[i]$ (for $1\leq i\leq n$) and $\myMatrix{b}$ are independently generated according to Bern($p_i$) and Bern($\frac{1}{2}$) respectively, where 
  $p_i \geq min(\frac{1}{2}, \frac{16}{H^{-1}(\delta)}\cdot \frac{\log_2 i}{i})$ for $\delta = \frac{i}{i+\log_2 k}$, and for $1\leq i\leq n-1$, $p_i\geq p_{i+1}, p_i\in (0,\frac{1}{2}]$. 
\end{definition}

It is worth observing that $\Hrennes{}$ marks a significant departure from prior families in the behavior of the density dependent on rows of the matrix $\myMatrix{A}$.  The sparsity of $\Hrennes{}$ is discussed in detail Section~\ref{sec:opt} showing that for even small $i$, $p_i$ can be set to values significantly smaller than $\frac{1}{2}$. 

\begin{restatable}{theorem}{conchash}
  \label{thm:main-rennes}
	 For $1\leq  m \leq n$, let  $h \xleftarrow{R} \Hrennes{k}$, $S \subseteq \{0,1\}^{n}$, $\Cell{S}{m} = \{y\in S \mid h^{(m)}(y) = \alpha^{(m)} \}$, $|S| \leq 2^m k$ for some $\alpha\in \{0,1\}^m$. Then for every value of $k>1$ and $\rho>1$, there exists $\qs\leq n$ such that for all $m$ with $\qs\leq m\leq n$, we have 
	\begin{align}
	E[\Cnt{S}{m}] = \frac{|S|}{2^m} \\
	\frac{\sigma^2[\Cnt{S}{m}]}{E[\Cnt{S}{m}]} \leq \rho
	\end{align} 
\end{restatable}
\begin{corollary}
$\Hrennes{k}$ is prefix-$(\rho,\qs,k)$-concentrated.
  \end{corollary}

The proof begins with the expression stated in Theorem~\ref{thm:main-dispexpression} and is based on analysis of dispersion index by considering separate cases for different sets of values of $w$. The case analysis especially for large values of $w$ turns out to be rather technical and uses the properties of distribution of binomial coefficients and Taylor expansion of $r(w,m)$, as detailed in Section~\ref{sec:rennes}.

\subsection{Approximate Model Counting using Concentrated Hashing}
As noted in Section~\ref{sec:prelims}, the usage of $(\rho,\qs, k)$-concentrated family does present the challenge of identification of application domains where such hash functions suffice. Typical usage of hash functions does not put restrictions on the size of the underlying set $S$ whose elements are being hashed. For example, the standard proofs of hashing-based counting techniques employ hash functions in the context where there is no reasonable upper bound on $|S|$.  Therefore, one wonders whether it is possible to design hashing-based counting techniques which can use concentrated hash functions without assuming an upper bound on $|S|$. 

We answer the above question positively in the third and final technical contribution of this paper with the design of approximate model counter with rigorous $(\varepsilon,\delta)$ guarantees {\SparseScalMC}, which employs a prefix $(\rho, \qs, \pivot)$-concentrated hash family instead of a strongly 2-universal hash family. 
\begin{restatable}{theorem}{scalmcthm}
    \label{thm:main-scalmc}
	For input formula $F$, tolerance parameter $\varepsilon$, confidence parameter $\delta$, and concentrated hashing parameters $\rho$ and $\qs$,  suppose ${\SparseScalMC}(F,\varepsilon,\delta,\rho)$ uses a prefix $(\rho, \qs,  \pivot)$-concentrated hash family  with the value of $\pivot = 78.72\cdot \rho(1+\frac{1}{\varepsilon})^2$ and returns an estimate $c$. 
Then,  
		$\Pr\left[\frac{|\satisfying{F}|}{1+\varepsilon} \leq c \right.$ $\left. \leq
	(1+\varepsilon)|\satisfying{F}|\right] \geq 1- \delta$. Furthermore, {\SparseScalMC} makes $ \mathcal{O}(2^{\qs+3} + \frac{\log(n) \log (1/\delta)}{\varepsilon^2})$ calls to a SAT-oracle. 
\end{restatable}

{\SparseScalMC} builds on the earlier algorithm \ScalApproxMC~\cite{CMV16,SM19}, but differs in the crucial use of a sparse hash family instead of a 2-universal hash family. This essentially requires us to rework the entire theoretical guarantees, which we do in Section~\ref{sec:sparsescalmc}.

Finally, in Section~\ref{sec:evaluation}, we evaluate the performance of {\SparseScalMC} using the sparse hash functions belonging to prefix $(1.1, 1,  \pivot)$-concentrated hash family and demonstrate that it leads to significant speedup in runtime over {\ScalApproxMC}. To the best of our knowledge,  this work is the first study to demonstrate runtime improvement using sparse hash functions without loss of $(\varepsilon,\delta)-$guarantees.

\section{Bounding the dispersion index}
\label{sec:concentrated}

In this section, we prove Theorem~\ref{thm:main-dispexpression}. Recall that for $1\leq i\leq n-1$, $p_i\in(0,\frac{1}{2}]$, $\mathcal{H}_{\{p_i\}} (n,n) \triangleq \{ h:\{0,1\}^{n} \rightarrow \{0,1\}^n \}$ denotes the family of functions of the form $h(x) = \myMatrix{A}x+\myMatrix{b}$ with $\myMatrix{A} \in \mathbb{F}_{2}^{n \times n}$ and $\myMatrix{b} \in \mathbb{F}_{2}^{n \times 1}$ where the entries of $\myMatrix{A}[i]$ (for $1\leq i\leq n$) and $\myMatrix{b}$ are independently generated according to Bern($p_i$) and Bern($\frac{1}{2}$) respectively.
Our first step is to compute the mean and bound the variance of $\Cnt{S}{m}$. We start with a known result and a definition.

\begin{lemma}~\cite{M99,AD16}\label{lem:prob}
	For all $\tau \in \{0,1\}^n$, we have
	\begin{align*} 
	  Pr (\myMatrix{A}^{(m)}\tau= \mathbf{0})= q(w,m)
	\end{align*}
	where $w=w(\tau)$ is the Hamming weight of $\tau$ (note that $0^0=1$).
\end{lemma}
\begin{proof}
	Since all the entries of b are chosen randomly with Bern($\frac{1}{2})$, for $y \in \{0,1\}^n$, we have $\prob[h^{(m)}(y) = \alpha^{(m)}] = \frac{1}{2^m}$, from which the expression for expectation follows. Now, for the variance we have
	$\sigma^2[\Cnt{S}{m}]  = 	  \sum_{y_1,y_2\in S}Pr [ h^{(m)}(y_1)=\alpha^{(m)}, h^{(m)}(y_2)=\alpha^{(m)}] - (\sum_{y \in S} Pr [h^{(m)}(y) = \alpha^{(m)}])^2$.
	
	\begin{align*}
	& \sum_{y_1,y_2\in S}Pr [ h^{(m)}(y_1)=\alpha^{(m)}, h^{(m)}(y_2)=\alpha^{(m)}] \\
	&= \sum_{y_1,y_2\in S}Pr [ h^{(m)}(y_1)=\alpha^{(m)} |  h^{(m)}(y_2)=\alpha^{(m)}] Pr [ h^{(m)}(y_2)=\alpha^{(m)}] \\
	&= \frac{1}{2^m} \sum_{y_1,y_2\in S} Pr [\myMatrix{A}^{(m)}y_1+b=\alpha^{(m)} |  \myMatrix{A}^{(m)}y_2+v=\alpha^{(m)}]\\
	&= \frac{1}{2^m} \sum_{y_1,y_2\in S}Pr[\myMatrix{A}^{(m)}(y_1-y_2)=0]
	\end{align*}
	
	where the randomness is over the choice of $\myMatrix{A}^{(m)}$. Now, $Pr[\myMatrix{A}^{(m)}(y_1-y_2)=0]$ depends on the Hamming weight $w$ of $y_1-y_2$ and is exactly the probability that the $w$ columns of $\myMatrix{A}^{(m)}$ corresponding to the bits in which $y_1$ and $y_2$ differ sum up to $0$ (mod 2). That is, 
	\begin{align*}
	\sum_{y_1,y_2\in S}Pr [ h^{(m)}(y_1)=\alpha^{(m)}, h^{(m)}(y_2)=\alpha^{(m)}] \\= 2^{-m} \sum_{x\in S}\sum_{w=0}^n c_S(w,x)q(w,m)
	\end{align*}
	where $c_S(w,x)$ is the number of vectors in $S$ that are at a Hamming distance of $w$ from $x$.
\end{proof}

We define $c_S(w,x) = |\{ y \mid y \in S, d(x,y) = w \}|$, i.e., the number of vectors in $S$ that are at a Hamming distance of $w$ from $x$. We also define $c_S(w) = |\{(x,y) \mid x \in S, y \in S, d(x,y)=w \}|$, i.e., the number of pairs of vectors in $S$ that are at Hamming distance $w$ from each other. Then we immediately obtain the following proposition (see~Appendix for details).
\begin{restatable}{proposition}{meansum}
\label{prop:mean-sum}The following expressions hold:
\begin{enumerate}
		\item 	$\expect[\Cnt{S}{m}] = \frac{|S|}{2^m}$
	        \item $	\sum_{y_1,y_2\in S}Pr [ h^{(m)}(y_1)= h^{(m)}(y_2)=\alpha^{(m)}] \\= 2^{-m} \sum_{x\in S}\sum_{w=0}^n c_S(w,x)q(w,m)$
\end{enumerate}
\end{restatable}

Then, we may express the variance in terms of $c_S(w)$ and $r(w,m)$.
\begin{lemma}\label{lem:variance}
	$\sigma^2[\Cnt{S}{m}] = \frac{ \sum_{w=0}^n c_S(w)r(w,m)}{2^m}$	
\end{lemma}
\begin{proof} $\sigma^2[\Cnt{S}{m}]=  \sum_{y_1,y_2\in S}Pr [ h^{(m)}(y_1)= h^{(m)}(y_2)=\alpha^{(m)}]-(\sum_{y \in S} Pr [h^{(m)}(y) = \alpha^{(m)}])^2$
\begin{align}
\begin{split}
  &=  2^{-m} \sum_{x\in S}\sum_{w=0}^n c_S(w,x)q(w,m) -\sum_{x\in S}\sum_{y\in S} \frac{1}{2^{2m}}\\
  &=  2^{-m} \sum_{x\in S}\sum_{w=0}^n c_S(w,x)q(w,m) -\sum_{x\in S}{\sum_{w=0}^n} \frac{c_S(w,x)}{2^{2m}}\\
  &=  2^{-m} \sum_{x\in S}\sum_{w=0}^n c_S(w,x)r(w,m) \label{eq:variance}
\end{split}
\end{align}
Earlier works on bounding $\sigma^2$ observed that $c_S(w,x) \leq {n \choose w}$ and focused their efforts to bound the resulting expression. Interestingly, the following seemingly simple rewriting allows us to explore interesting bounds for $\sigma^2$. We rewrite Eq~\ref{eq:variance} as 

\begin{align}\label{eq:variance-new-eq}
  \sigma^2[\Cnt{S}{m}] =2^{-m} \sum_{w=0}^n c_S(w)r(w,m)
\end{align}
where $c_S(w)$ is the number of pairs of vectors in $S$ that are at Hamming distance $w$ from each other.
\end{proof}

Next for all $m\in \{1,\ldots, n\}$ and every $S\subseteq \{0,1\}^n$ we use deep results from boolean functional analysis to bound the dispersion index, $\frac{\sigma^2[\Cnt{S}{m}]}{\expect[\Cnt{S}{m}]}$,  as a function of $|S|$ and $r(w,m)$. We start by setting up some notation. For $x,y \in \{0,1\}^n$, we say $y \subseteq x$ whenever for all $i\in [n]$, $y_i=1\implies x_i=1$. We say $S \subseteq \{0,1\}^n$ is a {\em down-set} if for all $x,y\in \{0,1\}^n$, $x \in S, y\subseteq x$ implies $y \in S$.  We say $S$ is {\em left-compressed } if, for all $x,y\in \{0,1\}^n$, $x \in S$ implies $y \in S$ whenever $y$ satisfies the two conditions (1) $|x| = |y|$ and (2) $x \succcurlyeq_{lex} y$, i.e., $x$ is lexicographically larger than $y$. For example, the set $\{000,001,100\}$ is a downset but it is not left compressed, while $\{000,001,010\}$ is both a downset and left-compressed.

In~\cite{R18}, it is shown that among all sets $S$ of the same cardinality, for all $k\in [n]$, $\sum_{w=0}^k c_{S}(w)$ achieves its maximum value for some left-compressed and down set $S$. We extend this to obtain the following crucial lemma.

\begin{restatable}{lemma}{downleftset}\label{sec:down-left-set}
	Let $n$ be positive integer and and let $t:[n]\rightarrow \mathbb{R}^+$ be a monotonically non-increasing function. Among all subsets $S$ of $\{0,1\}^n$ of same cardinality, the sum $\sum_{w=0}^n c_{S}(w)t(w)$ achieves its maximum value for some left-compressed and down set $S$.
\end{restatable}

\newcommand{\nbr}{\mathrm{nbr}}

The above lemma allows us to use the expressions obtained for $c_{S}(w)$ by Rashtchian in ~\cite{R18,RR19}. 

\begin{restatable}{lemma}{dispersion}\label{lm:disperson}~\cite{R18,RR19}
	For a left-compressed and down set $S$,	 $c_{S}(w) \leq  2\cdot \left(\frac{8e \sqrt{n\cdot\ell}}{w}\right)^w\cdot |S|$ where $\ell = \lceil \log |S| \rceil$. 
\end{restatable}

\begin{proof}
  The proof is based on the bounds derived  by Rashtchian in ~\cite{R18}. We give a few more details as  we will need them later when we explain our implementation.  More specifically, the proof uses Equations 4.2, 4.5, 4.8, and 4.10 from~\cite{R18}.  It is crucial to note that these equations hold only for a left-compressed and down set and not for an arbitrary set $S$. The proof follows by breaking into two cases based on the parity of $w$.

  For even $w$, Rashtchian upper bounds the expressions obtained in Eq. 4.2 and 4.5 by Eq 4.8 in ~\cite{R18}.  We rewrite Eq 4.8 by substituting $2t$ by $w$ to obtain $c_{S}(w) \leq 2\cdot \left(\frac{8e \sqrt{n\cdot\ell}}{w}\right)^w$. For odd $w$,  Rashtchian upper bounds the upper bound for $c_{S}(w)$ obtained in Eq. 4.2 and 4.5 by Eq 4.10. We rewrite Eq 4.10 by noting that $w = 2t+1$  to obtain $c_{S}(w) \leq 2 \cdot \left(\frac{8e }{w}\right)^w (\sqrt{n\cdot\ell})^{(w-1)} \ell $. Noting that $\ell \leq \sqrt{n\cdot \ell}$, we have $c_{S}(w) \leq  2 \cdot \left(\frac{8e \sqrt{n\cdot\ell}}{w}\right)^w$.  Thus, combining these cases, we get our lemma.
\end{proof}

 Thus, for any $S \subseteq \{0,1\}^n$ let us fix $\ell=\lceil{\log|S|\rceil}$. Then, \\
   $\frac{\sigma^2[\Cnt{S}{m}]}{\expect[\Cnt{S}{m}]} \leq 
 \sum\limits_{w=0}^{\ell} 2\cdot \left(\frac{8e \sqrt{n\cdot\ell}}{w}\right)^w r(w,m)$, which completes the proof of our first main result, Theorem~\ref{thm:main-dispexpression}, i.e.,

 \mainexp*

This theorem gives a closed form expression for upper bound on dispersion index, which is amenable to numerical computations. In particular, given $\ell$, one can compute the value of $p_i$'s such that dispersion index is upper bounded by a constant. Next, we analyze the behavior of $p_i$'s for a given upper bound on dispersion index and we construct concentrated hash functions based on their behavior. 

\section{A concentrated hash family}
\label{sec:rennes}
In this section, we finally construct a family of concentrated hash functions, which proves our second main Theorem~\ref{thm:main-rennes}, which we restate below.

\conchash*
\begin{proof}
	
  The first equation follows from Proposition~\ref{prop:mean-sum}. For the second, from Theorem~\ref{thm:main-dispexpression} we have, for any $1\leq m\leq n$,
  \begin{align*}
    \frac{\sigma^2[\Cnt{S}{m}]}{E[\Cnt{S}{m}]}& = 1 + \sum_{w=1}^{\ell} c_S (w) r(w,m) \\&\leq 1+ \sum_{w=1}^{\ell} 2\cdot \left( \frac{8e\sqrt{n\ell}}{w}\right)^w r(w,m), \\ \text{ where $\ell = \lceil m+\log_2(k) \rceil $. }
  \end{align*}
Note that $\frac{m}{\delta}+1\geq \ell\geq \frac{m}{\delta}$. Note that 
\begin{align*}
q(w,m) = \prod_{j=1}^{m}\left(\frac{1}{2}+\frac{1}{2}(1-2p_j)^w\right) \\ \leq \left(\frac{1}{2}+\frac{1}{2}(1-2p_m)^w\right)^m
\end{align*}

Now let us define $f(w) = \left(\frac{8e\sqrt{n\ell}}{w}\right)^wr(w,m)$. Then, we can divide into three cases: 
  \begin{caseof}
  \case{$1 \leq w \leq (2p_m)^{-1}$} 
{
  We have $\log(r(w,m))\leq - mwp_{m}$. To see this, following the reasoning from~\cite{AD16}, we have when $w\leq \frac{1}{2p_{m}}$,
  \begin{align*}
    \log(r(w,m))&\leq -m+m \log (1+ (1-2p_{m})^w)\\
    &\leq -m +m \log(1+e^{-2p_{m}w})\\
    &\leq -m +m(1-p_{m}w)
  \end{align*}
  where the last inequality follows from the fact that $\log_2(1+e^{-x})\leq 1-\frac{1}{2}x$ for $0\leq x\leq 1$ and that $0\leq 2p_{m}w\leq 1$ in this interval.
  Thus
\begin{align}
  \log_2 f(w) \leq w\log 8e\sqrt{n\ell}-w\log w -mp_{m}w  
\end{align}

Since $H^{-1}(\delta)\leq \delta/2$, we have $p_{m}\geq \frac{16}{H^{-1}(\delta)}\frac{\log m}{m}\geq \frac{32}{\delta} \frac{\log m}{m}\geq 32\frac{\log m}{m}$, since $\delta \leq 1$. Therefore
\begin{align*}
  &\log_2 f(w) \leq w\log 8e\sqrt{n\ell}-w\log w -mp_{m}w  \\
  &\leq w\log 8e\sqrt{n\ell}-mp_{m}w  \leq w\log 8e \sqrt{n\ell} -32w\log m\\
  &\leq w\log \frac{8e\sqrt{nl}}{m^{32}}
\end{align*}

Now, we pick $\qs_1> (\sqrt{n}\frac{2\rho}{\rho-1}(\ell)^{3/2} 8e)^{1/32}$. Note that this is possible, since $\ell\leq n$ and it suffices to choose $\qs_1> 1.12 (\frac{2\rho}{\rho-1})^{1/32}n^{1/16}$ which is in turn possible for any value of $\rho>1$.

Then, we have for any $m\geq \qs_1$, $m^{32}> 8e\sqrt{n} (\frac{2\rho}{\rho-1})\ell^{3/2}$ which implies that $\frac{8e\sqrt{nl}}{m^{32}}<\frac{\rho-1}{2\rho\ell}$. Then we have 
\begin{align}
  \log_2 f(w) &\leq w\log \frac{8e\sqrt{nl}}{m^{32}}\leq w\log \frac{\rho-1}{2\rho\ell}\leq 1\cdot \log \frac{\rho-1}{2\rho\ell}
    \end{align}
where the last inequality follows because, $\ell \geq 1$ (since $|S|\geq 2$), which means that $\log \frac{\rho-1}{2\rho\ell}<0$.

Therefore, $\sum_{w=1}^{\ell} f(w) \leq \frac{\ell(\rho-1)}{2\rho\ell}$
\begin{align*}
  \implies \frac{\sigma^2[\Cnt{S}{m}]}{E[\Cnt{S}{m}]} \leq 1 +  2\frac{\rho-1}{2\rho}< 1+\rho - 1=\rho\\
\implies \text{for $k>1$, $\rho>1$, for $\qs_1\leq m\leq n$,}  \frac{\sigma^2[\Cnt{S}{m}]}{E[\Cnt{S}{m}]} < \rho
\end{align*}
}

\case{$(2p)^{-1} \leq w \leq \frac{m H^{-1}(\delta)}{16}$}
     {  We start by observing that $g(w)=\log f(w)=w \log 8e\sqrt{n\ell} - w\log w$ is increasing in the interval $w=0$ to $w=\ell$. To see this, consider the derivative $g'(w)=\log(\frac{8e\sqrt{n\ell}}{w})-1$. Then $w\leq \ell\leq 8e\sqrt{n\ell}$ implies $2\leq \frac{8e\sqrt{n\ell}}{w}$, which implies $g'(w)>0$.

       Now $w\leq \frac{mH^{-1}(\delta)}{16}\leq \frac{m\delta}{32}\leq \frac{m}{32}$ since $\delta \leq 1$ and $H^{-1}(\delta)\leq \frac{\delta}{2}$. Thus we have,     
       \begin{align*}
	 \log f(w) &\leq \log f(\frac{m}{32})\leq \frac{m}{32} \log (\frac{2^9 e\sqrt{n\ell}}{m})\\
         &= \frac{9m}{32}+m \log \left(\frac{e\sqrt{n\ell}}{m}\right)^{\frac{1}{32}}
       \end{align*}
Now we pick $m>\frac{e\sqrt{n\ell}}{40}$. Then we get $\log f(w)         \leq 0.282m+0.167m \leq 0.45m$.

On the other hand, we have $\log r(w,m)\leq -m +m \log (1+exp(-2p_{m}w))$ $\leq -m+m\log(1+exp(-1))\leq -0.58m$

Thus, we get %
   $\frac{\sigma^2[\Cnt{S}{m}]}{E[\Cnt{S}{m}]} \leq 1+ \sum_{w=1}^\ell 2\cdot f(w)r(w,m) \leq 1+\sum_{w=1}^\ell 2^{1-0.13m}\leq 1+\frac{2\ell}{2^{0.13m}}$
Thus there exists  $\qs_2 =  \frac{1}{0.13}\log_2 \left(\frac{2\ell}{\rho -1}\right)$ such that for $\qs_2\leq m\leq n$, clearly this can be made less than any constant $\rho>1$.
}
\case{$w \geq \frac{m H^{-1}(\delta)}{16}$. We start with a claim, the proof for which can be found in the Appendix.}
     {%
    \begin{restatable}{claim}{rbound}\label{lm:rbound}
    		For $m \geq 2$, if $w \geq \frac{m H^{-1}(\delta)}{16}$, then $\log_2 r(w,m) < -m + 1 -\log_2 m$
    \end{restatable}

From the above claim, we have %
$\log_2 r(w,m) \leq -m + 1 - \log_2 m$, i.e., $ r(w,m) \leq \frac{2\cdot 2^{-m}}{m}$.

Also, recalling that we have $\Sigma_{w=1}^\ell c_s(w)\leq 2^\ell$
,  we obtain 
$\frac{\sigma^2}{\mu}\leq 1+ \sum_{w=1}^\ell c_s(w) max_w r(w,m) \leq 1+ 2^\ell \cdot \frac{2\cdot 2^{-m}}{m} = 1+\frac{2k}{m}$

Thus for all $k$, we can pick $\qs_3 = \frac{2k}{\rho -1}$  such that for any $\qs_3\leq m\leq n$ and $\rho>1$, $\frac{\sigma^2}{\mu}\leq \rho$.

     }\label{case:largew}
\end{caseof}
Combining the three cases and taking $\qs=\max\{\qs_1,\qs_2,\qs_3\}$, we obtain our desired result.
It is worth noting that the smallest value of $m$ (i.e., $\qs$) for which Theorem~\ref{thm:main-rennes} holds true depends on $\rho$ and $k$. Furthermore, it is interesting to observe that the proof of Case 3 crucially depends on usage of $r(w,m)$ instead of $q(w,m)$ in the expression of  $\frac{\sigma^2}{\mu}$ as the current proof techniques would only yield  $\log_2 q(w,m) < -m + 1$, which would be insufficient to prove $\frac{\sigma^2}{\mu}\leq \rho$.  
\end{proof}
\section{A New Approximate Model Counting Algorithm: {\SparseScalMC}}\label{sec:sparsescalmc}
In this section, we seek to design algorithms that can use $(\rho, \qs, k)$-concentrated hash functions for a small $k$, independent of the problem instance.  In particular, we first revisit the state of the art approximate counting algorithm {\ApproxMC}. We will refer to the algorithmic constructs presented in {\ApproxMCTwo}~\cite{CMV16} since the subsequent versions, i.e., {\ApproxMCThree} and {\ScalApproxMC}, have proposed algorithmic improvement to the underlying {\SAT} calls only.  We seek to  modify {\ApproxMCTwo} so as to employ concentrated hash function; the final implementation of {\SparseScalMC} builds on top of {\ScalApproxMC}, allowing it to benefit from the improvements proposed in {\ApproxMCThree} and {\ScalApproxMC}. 

\subsection{The Algorithm}

The subroutine {\SparseScalMC} is presented in Algorithm~\ref{alg:scalmc}.  
 {\SparseScalMC} takes in a formula $F$, tolerance: $\varepsilon$, and confidence parameter $\delta$, concentrated hashing parameters $\rho$ and $\qs$ as input and returns an estimate of $|\satisfying{F}|$ within tolerance $\varepsilon$ and confidence at least $1-\delta$.   Similar to {\ApproxMCTwo}, the key idea of {\SparseScalMC} is to partition the solution space of $F$ into {\em roughly equal small} cells of solutions such that the $|\satisfying{F}|$ can be estimated from the number of solutions in a randomly chosen cell scaled by the total number of cells. This idea requires two crucial ingredients:

\begin{enumerate}
	\item hash functions to achieve desired properties of partitioning:	As has been emphasized earlier, in this work, we mark a departure from prior work and employ concentrated hash functions instead of strongly 2-universal hash functions. 

	\item subroutine to check whether a cell is small, i.e., the number of solutions in the cell is less than an appropriately computed $\thresh$. {\SparseScalMC}
	 assumes access to the subroutine $\BoundedSAT$ that takes in a formula $F$ and a threshold $\thresh$ and returns an integer $Y$, such that  $Y = \min(\thresh, |\satisfying{F}|)$. Note that $Y = \thresh$ is used to indicate that the number of solutions is greater than or equal to $\thresh$, which indicates that the cell is not {\em small}. We do not treat {\BoundedSAT} as an oracle in our analysis and instead as a subroutine which uses a NP oracle to enumerate solutions of $F$ one by one until we have found $\thresh$ number of solutions or there are no more solutions. As such for {\BoundedSAT} to make polynomially many calls to NP oracle, {\thresh} is polynomial in $\frac{1}{\varepsilon}$. 
	 \item  Subroutine, called {\FibBinSearch}, to search for the right number of cells as discussed in detail below. 
\end{enumerate}

{\SparseScalMC} differs from {\ApproxMCTwo} primarily in the computation of $\thresh$ and usage of concentrated hash functions -- the two critical components that distinguish several hashing-based counting techniques.  The computation of $\thresh$ involves the parameter $\rho$ to account for concentrated hashing and incurs an overhead proportional to $\rho$. As discussed later, for our empirical studies, we set $\rho$ to 1.1. Unlike prior techniques, we introduce another parameter $\iniThresh$ that depends on $\thresh$ and $\qs$ to account for $\qs$ parameter of concentrated hash functions. 
 {\SparseScalMC} first checks if the number of solutions of $F$ is less than $\iniThresh$ and upon passing the check it simply returns the number of solutions of $F$. For interesting instances, the check fails and  {\SparseScalMC} invokes the subroutine {\SparseScalMCCore} $t$ times and computes the median of the returned estimates by {\SparseScalMCCore}.

 The subroutine {\SparseScalMCCore} lies at the core of {\SparseScalMC} and shares similarity with {\ApproxMCTwoCore} employed in~\cite{CMV16}. In contrast to {\ApproxMCTwoCore}, the algorithmic description does not restrict the hash family to $H_{xor}$ in line~\ref{line:sfc-choose-alpha}. We use $\mathcal{H}_{\rho, \qs}(n,n)$ as a placeholder for a hash family, whose properties would be inferred from the analysis of {\SparseScalMC} and stated formally in Lemma~\ref{lm:aux-bounds}.

  {\SparseScalMCCore} takes in a formula $F$, $\thresh$, and returns $\solCount$  as an estimate of $|\satisfying{F}|$ within tolerance $\varepsilon$ corresponding to $\thresh$. To this end, {\SparseScalMCCore} first chooses a hash function $h$ from a prefix-family $\mathcal{H}_{\rho, \qs}(n,n)$ and a cell $\alpha$. As noted above, we use {\em prefix-slices} of $h$ and $\alpha$. After choosing $h$ and $\alpha$ randomly,
 {\SparseScalMCCore} checks if
 $\Cnt{F}{n} < \hiThresh$.  If not,
 {\SparseScalMCCore} fails and returns $2^n$.(A careful reader would note that we could have chosen any arbitrary number to return) 
 Otherwise, it invokes sub-routine {\FibBinSearch} to find a value of
 $m$ (and hence, of $h^{(m)}$ and $\alpha^{(m)}$) such that
 $\Cnt{F}{m} < \hiThresh$
 and $\Cnt{F}{m-1} \ge
 \hiThresh$. The reason behind the particular choice of the value of $m$ is that to obtain higher confidence in the counts returned by {\SparseScalMCCore}, we would ideally like the $\expect[\Cnt{F}{m} ]$ to be high so as to obtain better bounds through concentration inequalities. Of course, we can only handle the cases when $\Cnt{F}{m}  $  is polynomial to ensure polynomially many calls to NP oracle (SAT solver in practice).  
 The implementation of {\FibBinSearch} is provided in ~\cite{CMV16} and we use the procedure as-is. The invocation of {\BoundedSAT} in line~\ref{line:sfc-bsat-m} calculates
 $\Cnt{F}{m}$.
 Finally,
 {\SparseScalMCCore} returns $(2^m \times {\Cnt{F}{m}})$, where $2^m$ is the number of cells that $\satisfying{F}$ is partitioned into by $h^{(m)}$. 
\begin{algorithm}
	\caption{$\SparseScalMC(F,\varepsilon,\delta,\rho,\qs)$}
	\label{alg:scalmc}
	\begin{algorithmic}[1]
		\State $\hiThresh \gets 1+ 9.84\cdot \rho \cdot \left(1 + \frac{\varepsilon}{1+\varepsilon}\right)\left(1 + \frac{1}{\varepsilon}\right)^2$; 
		\State $\iniThresh \gets \hiThresh * 2^{\qs+3}$
		\State $Y \gets \BoundedSAT(F, \iniThresh)$;  \label{line:scalmc-boundedsat}
		\If {($Y < \iniThresh$)} \Return $|Y|$;
		\EndIf
		\State $t \leftarrow  \lceil17\log_2 (3/\delta)\rceil$;
		\State $\cellCount \gets 2$; $C \gets \emptyList$; $\iter \gets 0$;
		\Repeat
		\State $\iter \gets \iter+1$;
		\State  $ \solCount \gets \SparseScalMCCore(F,\rho ,\hiThresh)$;
	 \State $\AddToList(C,\solCount)$;
		\Until {($\iter < t$)};
		\State $\mathrm{finalEstimate} \gets \FindMedian(C)$;
		\State \Return $\mathrm{finalEstimate}$
	\end{algorithmic}
\end{algorithm}

\begin{algorithm}
	\caption{$\SparseScalMCCore(F,\rho, \hiThresh)$}
	\label{alg:scalmccore}
	\begin{algorithmic}[1]
		\label{line:sfc-choose-h}
		\State Choose $h$ at random from $\mathcal{H}_{\rho}(n, n)$;
		\label{line:sfc-choose-alpha}
		\State Choose $\alpha$ at random from $\{0, 1\}^{n}$;
		\State $Y  \gets {\BoundedSAT}(F\wedge (h^{(n)})^{-1} (\alpha^{(n)}), \hiThresh)$;
		\If {($|Y| \ge \hiThresh$)} \Return $2^n$
		\EndIf
		\State $m \gets {\FibBinSearch}(F, h, \alpha, \hiThresh)$;
		\State $ \Cnt{F}{m} \gets {\BoundedSAT}(F\wedge (h^{(m)})^{-1} ( \alpha^{(m)}), \hiThresh)$; \label{line:sfc-bsat-m} 
		\State \Return $(2^m \times  \Cnt{F}{m})$;
	\end{algorithmic}
\end{algorithm}

\subsection{Analysis of {\SparseScalMC}}
We now present the analysis of {\SparseScalMC}. The primary purpose of this section is to highlight the sufficiency of concentrated hashing for the theoretical guarantees of {\SparseScalMC}. 

\newcommand{\good}{\mathsf{Good}}
\newcommand{\bad}{\mathsf{Bad}}

Let $\bad$ denote the event that {\SparseScalMCCore} either returns
$(\bot, \bot)$ or returns a pair $(2^m, \solCount)$ such that
$2^m\times\solCount$ does not lie in the interval 
$I_\good=\left[\frac{|\satisfying{F}|}{1+\varepsilon},
|\satisfying{F}|(1 + \varepsilon)\right]$.  We wish to
bound $\prob\left[\bad\right]$ from above. Towards this end, for $i\in \{1,\ldots,n\}$, let $T_{i}$ denote the event $\left(\Cnt{F}{i} < \hiThresh\right)$, and let $L_i$ and $U_i$ denote the events
$\left(\Cnt{F}{i} <
\frac{|\satisfying{F}|}{(1+\varepsilon)2^{i}}\right)$ and
$\left(\Cnt{F}{i} > \right.$ $\left.
\frac{|\satisfying{F}|}{2^i}(1+\frac{\varepsilon}{1+\varepsilon})\right)$,
respectively.

For any event $E$, let $\overline{E}$ denote its complement. Now, for $\bad$ to happen, \SparseScalMCCore\ must return (at some iteration $i$) with $L_i$ or $U_i$. Further, if it returned at $i$, then $T_i$ holds and $T_{i-1}$ must not hold (else it would have returned at iteration $i-1$ itself). 
Thus, we obtain %
\begin{align}
  \label{eq:bad-ub}
  \prob\left[\bad\right] \leq 
	\prob\left[\bigcup_{i \in \{1, \ldots n\}} \left(\overline{T_{i-1}}
	\cap T_i \cap (L_i \cup U_i)\right)\right]\end{align}
Note that we only get an upper bound (and not an equality) above because the interval $I_\good$ considered has upper bound $|\satisfying{F}|(1+\varepsilon)$, while $U_i$ and $\thresh$ are defined using the factor $(1+\frac{\varepsilon}{1+\varepsilon})\leq 1+\varepsilon$.

Our next goal is to simplify this upper bound. Let $m^*$ be the smallest $i$ such that $\frac{|\satisfying{F}|}{2^i}(1+\varepsilon)\leq \thresh-1$. This value must exist since $|\satisfying{F}|\geq \iniThresh$. Note that when $|\satisfying{F}|<\iniThresh$, the algorithm returns the exact count and hence is guaranteed to be correct.
        Now, by substituting the chosen value of $\thresh$ and simplifying, we obtain
        \begin{align} m^*=\left \lfloor \log_2 |\satisfying{F}| - \log_2 \left(4.92\cdot \rho \cdot \left(1 +\frac{1}{\varepsilon}\right)^2\right)\right\rfloor  \label{eq:mstardef}\end{align}

From the definition of $m^*$, we have $2^{m^*+1} \geq \frac{2*|\satisfying{F}|}{\thresh-1}$.  Since  $|\satisfying{F}| \geq \iniThresh$, we have $2^{m^*+1} \geq \frac{2*\thresh*2^{\qs+3}}{\thresh-1}$, i.e.,  $m^*+1 \geq \qs+4$, or $m^*-3 \geq \qs$.

Similar to {\ScalApproxMC}, we show that for {\SparseScalMC}, one can upper bound $\bad$ by considering only five events, namely, $T_{m^*-3} L_{m^*-2},$ $L_{m^*-1}, L_{m^*}$ and $U_{m^*}$. It is worth noting that  the proof only requires usage of prefix-hash family in the algorithm with no further restrictions on nature of the prefix-hash family. In fact, the main property that we need from the prefix hash family, which follows from Proposition~\ref{prop:prefix-monotonicity}, is that
          \begin{align}\label{eq:thresh_monotone}
            \forall j\in \{1,\ldots, n\}, T_j\implies T_{j+1}
            \end{align}

          \begin{restatable}{lemma}{upperbound}
            \label{lm:upperbound}
		$\prob[\bad] \leq \prob[T_{m^*-3}]
		+ \prob[L_{m^*-2}] + \prob[L_{m^*-1}] + \prob[L_{m^*} \cup U_{m^*}]$
\end{restatable}

 The following lemma utilizes the key property of concentrated hash families stated in Proposition ~\ref{prop:probBounds} to bound the probabilities of the concerned events. 

\begin{restatable}{lemma}{auxbounds}
    \label{lm:aux-bounds}
 If $\mathcal{H}$ is prefix-$(\rho, \qs, \mathrm{pivot})$-concentrated family for $\pivot =78.72\cdot \rho(1+\frac{1}{\varepsilon})^2$, then the following bounds hold:
\begin{enumerate} 

\item\label{lm:aux-bound:lm-um}  $\prob[L_{m^*} \cup U_{m^*}] \leq \frac{1}{4.92}$

 \item  \label{lm:aux-bound:lm1} 
 $\prob[L_{m^*-1}] \leq  \frac{1}{10.84}$ 

\item \label{lm:aux-bound:lm2} 
$\prob[L_{m^*-2}] \leq \frac{1}{20.68}$

\item \label{lm:aux-bound:lm3} 
$\prob[T_{m^*-3}] \leq \frac{1}{62.5}$ 

\end{enumerate}
\end{restatable}
\begin{proof}
	Note that $\prob[T_i] = \prob[\Cnt{F}{i} \leq \thresh]$ and  $\prob[L_i] = \prob\left[\Cnt{F}{i} \right.$ $\left. \leq  (1+\varepsilon)^{-1} \mu_{i}\right]$. Furthermore, \\ $\prob[L_i \cup U_i] = \prob\left[\Cnt{F}{i} - \mu_{i}|\geq \frac{\varepsilon}{1+\varepsilon} \mu_{i}\right] $
	To obtain bounds, we substitute  values of  $m^*$, $\hiThresh$, $\mu_i$, and we seek to apply  Proposition~\ref{prop:probBounds} with appropriate values of $\beta$. We observe that to obtain (~\ref{lm:aux-bound:lm-um}), it is sufficient to employ $(\rho, m^*, \frac{\pivot}{8})$ concentrated family; Similarly, to obtain (~\ref{lm:aux-bound:lm1}), (~\ref{lm:aux-bound:lm2}) , (~\ref{lm:aux-bound:lm3}), it is sufficient to employ $(\rho, m^*-1, \pivot/4)$, $(\rho,  m^*-2, \frac{\pivot}{2})$, $(\rho, m^*-3, \pivot)$ concentrated families respectively. Proposition~\ref{prop:conc-mono} allows us to conclude that $(\rho, m^*-3, \pivot)$  concentrated family  suffices to obtain the above bounds.  Since $m^*-3 \geq \qs$, we conclude that $(\rho, \qs,\pivot)$-concentrated family suffices to obtain the above bounds.	 
\end{proof}

Combining Lemma~\ref{lm:upperbound} with the observation that {\SparseScalMCCore} is invoked $\mathcal{O}(\log \frac{1}{\delta})$ times and we return median as the estimate, we obtain the following correctness and time complexity for {\SparseScalMC} by using the standard Chernoff analysis for the amplification of probability bounds. 

\scalmcthm*

The correctness and time complexity of {\SparseScalMC} have exactly the same expression as that of {\ScalApproxMC}. Theorem~\ref{thm:main-scalmc} highlights that prefix-$(\rho,  \qs, \pivot)-$concentrated hash family are sufficient to provide $(\varepsilon,\delta)$ estimates. In fact, in our experimental results that we discuss next, we will use a sparse hash function belonging to this family. %

\subsection{Further Optimizations }\label{sec:opt}

As mentioned earlier, {\BoundedSAT} is a subroutine that takes in a formula $F$ and threshold $\thresh$, and uses a NP oracle to enumerate solutiosn of $F$ one by one until we have found the desired threshold number of solutions or there are no more solutions. The practical implementation of {\BoundedSAT} replaces NP oracle with SAT solver and and as such for a fixed formula $F$, the runtime of {\BoundedSAT} depends on $\thresh$. The usage of $(\rho, \qs, \pivot)$-concentrated family leads to invocation of {\BoundedSAT} with threshold set to $\iniThresh$ in line~\ref{line:scalmc-boundedsat} of {\ScalApproxMC} algorithm. Therefore, for practical efficiency, it is desirable to construct concentrated families with as small values of $\qs$ as possible. The bound on $\qs$ provided by the proof of Theorem~\ref{thm:main-rennes} is prohibitively large ($\qs > 70$) even for $n=10$.  To this end, we turn to analytical techniques aided by scientific programming in Python. 

For given $n$, $k$, and $\rho$, we seek to compute as small values of $p_i$ as possible while satisfying $p_i \geq p_{i+1}$.  As a first step, we observe that the upper bound for $c_{S}(w)$ employed above is a loose upper bound and accordingly, 
the bounds on the constants $p_i$ (as well as $k$ and the large enough value of $m$) obtained from our analysis above are very loose.  To this end, we compute $c_{S}(w)$ based on the Eq 4.2 and Eq 4.5 obtained in~\cite{R18}, as indicated in the proof of Lemma~\ref{lm:disperson}. We then compute the values of $p_i$ for $\qs = 1$, $k = 512$, and $\rho=1.1$. The particular values for $\rho$ and $k$ were chosen due to their usage in experimental evaluation of {\SparseScalMC}. We call the resulting family ${\Hrennes{lsa}} $ and  employ ${\Hrennes{lsa}}$ in our empirical evaluation. 

Figure~\ref{fig:trend} plots the values of computed $p_i$ vis-a-vis $i$ We also plot another curve $f(i) = \frac{1.6 \log_2 (i+1)}{i}$. It is interesting to observe that the two curves fit nicely to each other. To illustrate the gap between observed and theoretical bound, we plot the bound on $p$ obtained from Theorem~\ref{thm:main-rennes} as $g(i) = \frac{32 \log_2 (i+1)}{\delta \cdot i}$ noting that $H^{-1}(\delta) < \frac{\delta}{2}$. 

\begin{figure}[htb]
	\includegraphics[width=0.8\linewidth]{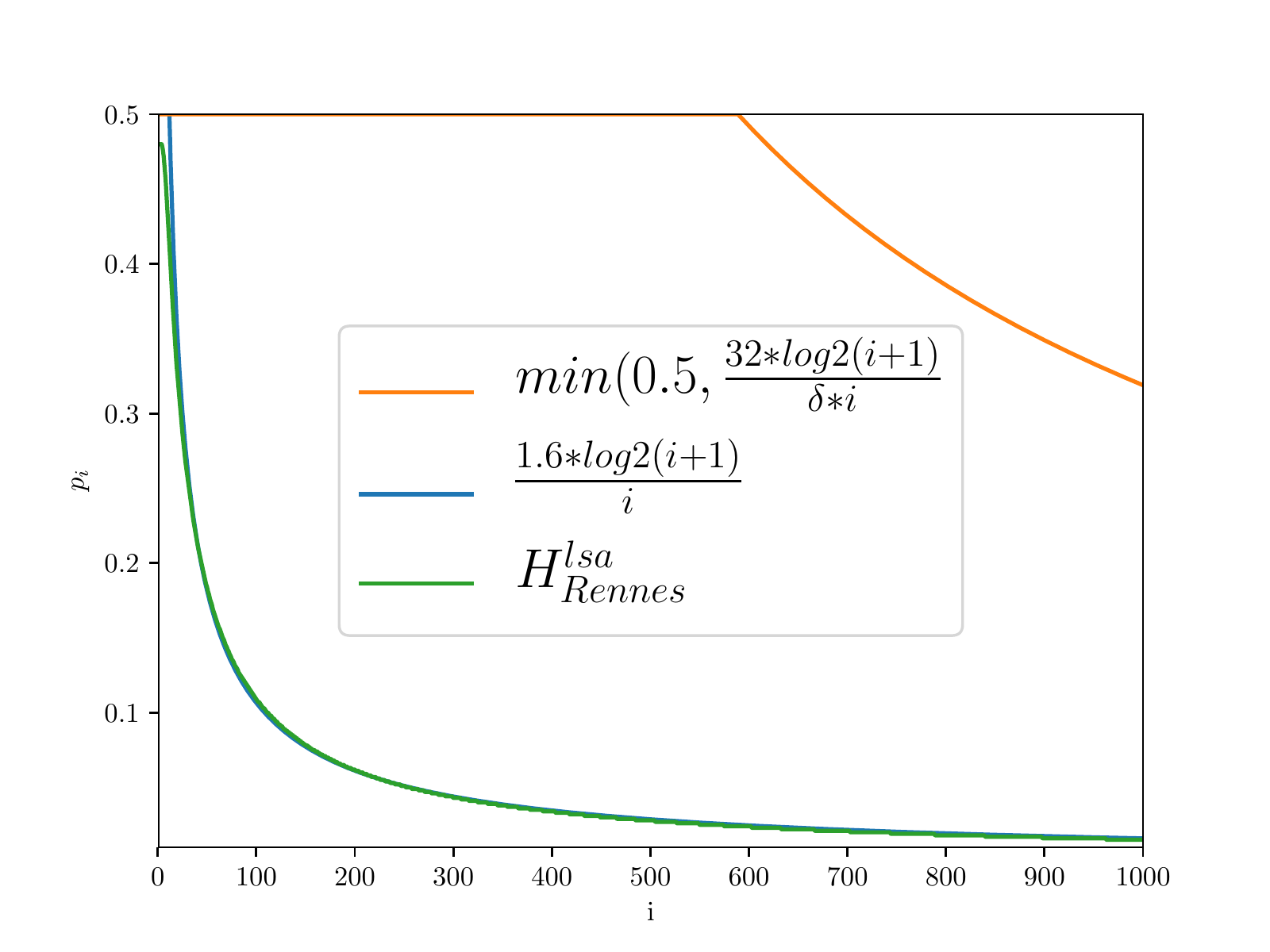}
	\caption{Trend of $p_i$ vis-a-vis $i$ }
	\label{fig:trend}
\end{figure}
The large difference between the two plots clearly illustrates the potential for improvement of constants in Theorem~\ref{thm:main-dispexpression} and we leave this as a natural direction of future work. Furthermore, we conjecture existence of sparse prefix hash functions with $p_m = \mathcal{O}(\frac{\log m}{m})$ belonging to  $(\rho,1, \kappa)$-concentrated family.

\begin{table*}[h]
	\scriptsize
	\centering
	\begin{tabular}{ |p{3cm}|c|c|c|c|c|c|c| }

		\hline \multicolumn{1}{|c|}{\textbf{Benchmark}} & \multicolumn{1}{c|}{\textbf{Vars}} & \multicolumn{1}{c|}{\textbf{Clauses}} & 
		\multicolumn{1}{|c|}{$|\mathcal{P}|$} &
		\multicolumn{1}{c|}{$\log_2$(Count)} & 
		\multicolumn{1}{c|}{\textbf{{\ScalApproxMC} time}} & 
		\multicolumn{1}{c|}{\textbf{{\SparseScalMC} time}} & \multicolumn{1}{c|}{\textbf{Speedup}} \\ \hline  
		
		\hline
		
		10B-1&15390&68337&174&56.17&4274.56&--&--\\\hline
or-100-20-7-UC-40&200&539&200&56.55&3526.45&--&--\\\hline		
		03B-4&27966&123568&114&28.55&983.72&1548.96&0.64\\\hline
		blasted\_TR\_b12\_2\_linear&2426&8373&107&63.93&32.07&56.75&0.57\\\hline
		blasted\_squaring23&710&2268&61&23.11&0.66&1.21&0.55\\\hline
		blasted\_case144&765&2340&138&82.07&102.65&202.06&0.51\\\hline
		
		modexp8-4-6&83953&316814&88&32.13&788.23&920.34&0.86\\\hline
		or-70-5-5-UC-20&140&360&140&43.91&675.1&788.74&0.86\\\hline

min-28s&3933&13118&464&459.23&48.63&35.83&1.36\\\hline
90-14-8-q&924&811&924&728.29&242.07&178.93&1.35\\\hline
s9234a\_7\_4&6313&14555&247&246.0&4.77&2.45&1.95\\\hline
min-8&1545&4230&288&284.78&8.86&4.59&1.93\\\hline
s13207a\_7\_4&9386&20635&700&699.0&34.94&17.05&2.05\\\hline
min-16&3065&8526&544&539.88&33.67&16.61&2.03\\\hline
90-15-4-q&1065&911&1065&839.25&273.1&135.75&2.01\\\hline

s35932\_15\_7&17918&44709&1763&1761.0&--&72.32&--\\\hline
s38417\_3\_2&25528&57586&1664&1663.02&--&71.04&--\\\hline

75-10-8-q&460&465&460&360.13&--&4850.28&--\\\hline
90-15-8-q&1065&951&1065&840.0&--&3717.05&--\\\hline

	\end{tabular}
	\caption{Runtime performance comparison of {\SparseScalMC} vis-a-vis {\ScalApproxMC}. (Timeout:  5000 seconds)}
	\label{tab:performance}
\end{table*}
\section{Experimental Evaluation}\label{sec:evaluation}
In this section, we evaluate the performance of our approximate model counting algorithm {\SparseScalMC} using the prefix $(1.1, 1,  \pivot)$-concentrated hash family {\Hrennes{lsa}} \footnote{Our theoretical analysis of {\SparseScalMC} allows all values of $\rho \geq 1$ and $\qs > 1$; we leave further optimization of the choice of $\rho$ as future work.}.  For all our experiments, we used $\varepsilon = 0.8$ and $\delta=0.1$, which is in line with the chosen values for these parameters in previous studies on counting. The setting of $\varepsilon = 0.8$ yields $\mathrm{pivot}$ to be $512$.  Recall that prior empirical studies had to sacrifice theoretical guarantees due to their reliance on far fewer invocations of SAT solver than those dictated by the theoretical analysis~\cite{EGSS14a,ZCSE16,AT17,AHT18}. In contrast, we use a faithful implementation of {\SparseScalMC} that retains theoretical guarantees of $(\varepsilon, \delta)$ approximation. {\SparseScalMC} is publicly available as an open source software at: \url{https://github.com/meelgroup/approxmc}. 

To evaluate the runtime performance and quality of approximations computed by {\SparseScalMC}, we conducted a comprehensive performance evaluation of counting algorithms involving 1896 benchmarks. Most practical applications of model counting reduce to projected counting and therefore, keeping in line with the prior work, we experiment with benchmarks arising from wide range of application areas including probabilistic reasoning, plan recognition, DQMR networks, ISCAS89 combinatorial circuits, quantified information flow, program synthesis, functional synthesis, logistics, as have been previously employed in studies on model counting~\cite{CMV16,LM17}.  We perform runtime comparisons with {\ScalApproxMC} as {\ScalApproxMC} was shown to be state of the art approximate counter with significant performance gain over other approximate counters~\cite{SGM20,SM19}. 

The objective of our experimental evaluation was to answer the following questions:
\begin{enumerate}
	\item How does runtime performance of {\SparseScalMC} compare with that of {\ScalApproxMC}? 
	\item How far are the counts computed by {\SparseScalMC} from the exact counts?
\end{enumerate}

The experiments were conducted on a high performance computer cluster, with each node consisting of an E5-2690 v3 CPU with 24 cores and 96GB of RAM such that each core's access was restricted to 4GB. The computational effort for the evaluation consisted of over 20,000 hours. We used timeout of 5,000 seconds for each experiment, which consisted of running a tool on a particular benchmark. To further optimize the running time for both {\ScalApproxMC} and {\SparseScalMC}, we used improved estimates of the iteration count $t$ following an analysis similar to that in~\cite{CMV16}. 

\subsection{Results}

\subsection*{Runtime performance}

We present the runtime comparison of {\SparseScalMC} vis-a-vis {\ScalApproxMC} in Table~\ref{tab:performance} on a subset of our benchmarks~\footnote{The entire set of benchmarks and the corresponding set of  logs generated by {\ScalApproxMC} and {\SparseScalMC} are available at  \url{https://doi.org/10.5281/zenodo.3766168}}. Column 1 specifies the name of the benchmark, while columns 2 and 3 list the number of variables and clauses, respectively. Column 4 Column 4 lists the $\log_2$ of the estimate returned by {\SparseScalMC}. Columns 5 and 6 list the runtime (in seconds) of {\SparseScalMC} and {\ScalApproxMC} respectively. Column 7 indicates speedup of {\SparseScalMC} over {\ScalApproxMC}. We observe the following:%
\begin{enumerate}
	\item {\SparseScalMC} significantly outperforms {\ScalApproxMC} for a large set of benchmarks. We observe that {\SparseScalMC} is able to compute estimates for formulas for which {\ScalApproxMC} timed out. Furthermore, {\SparseScalMC} is also significantly faster for most of the benchmarks where {\ScalApproxMC} does not timeout. 
	
	\item  Recall that the density of XORs decreases with increase in $\log_2 |\satisfying{F}|$ and we observe that the performance of {\SparseScalMC} too improves further as the number of solutions of $F$ increases. It is worth noting that for a subset of benchmarks, {\SparseScalMC} is slower than {\ScalApproxMC}.
\end{enumerate}
Upon further investigation, we observe a strong correlation between the speedup and the $\log_2$ of the number of solutions. It is worth recalling that the number of XORs required to ensure that a randomly chosen cell is small is close to $\log_2$ of the number of solutions. Since for a fixed number of variables, the sparsity increases with the number of XORs, there is a tradedoff between the gains due to sparse XORs over the increased overhead of requirement of enumerating higher number of solutions due to increased $\thresh$. 
It is worth viewing the runtime improvement in the context of prior work where significant slowdown was observed. 

\subsection*{Approximation Quality}

\begin{figure}[h]
	\centering
	\includegraphics[width=0.8\linewidth]{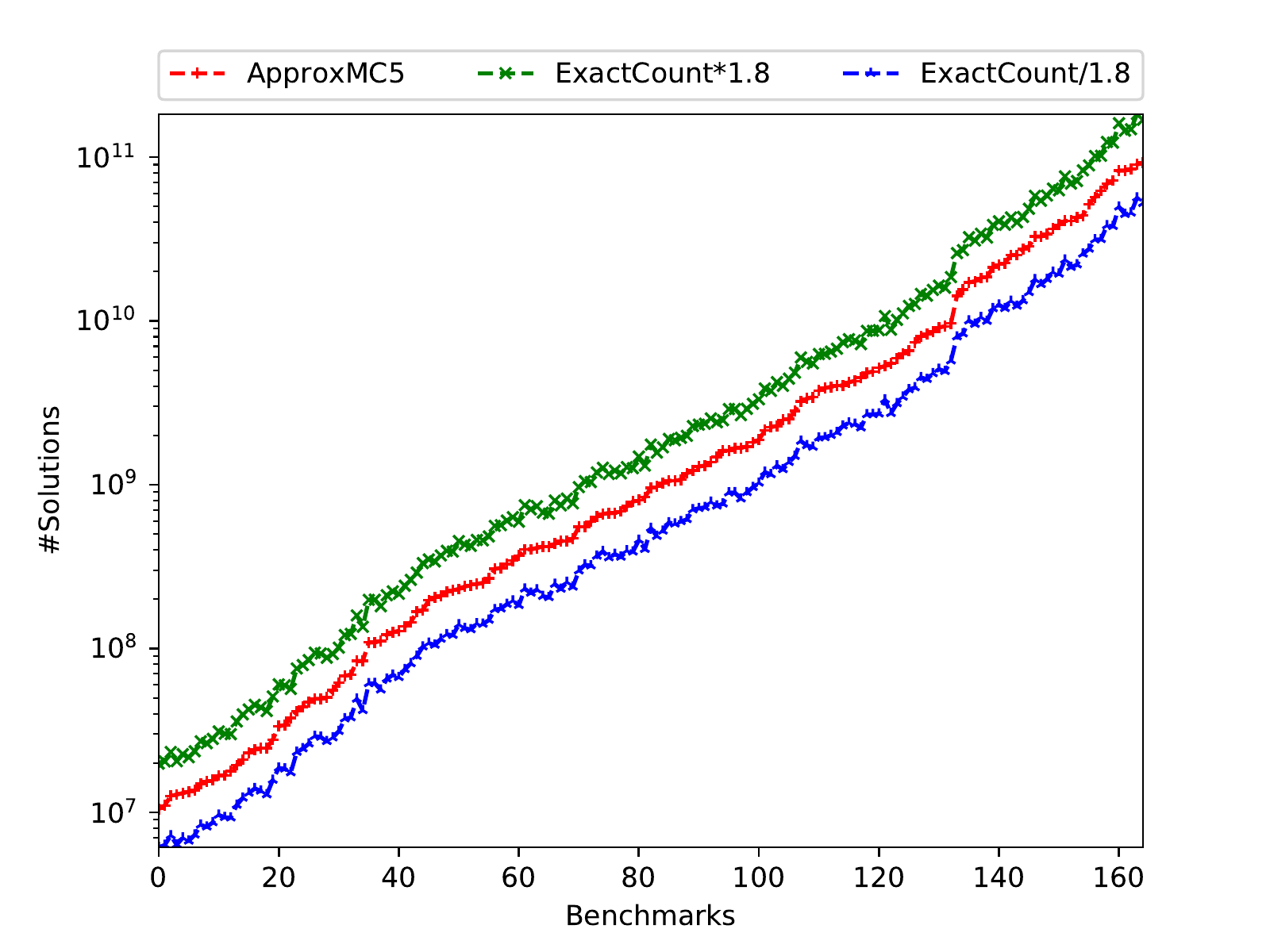}
	\caption{Plot showing counts obtained by {\SparseScalMC} vis-a-vis exact counts from {\DSharp} }
	\label{fig:quality}
\end{figure}

To measure the quality of approximation, we compared the approximate
counts returned by {\SparseScalMC} with the counts computed by an
exact model counter, viz. {\DSharp}.
Figure~\ref{fig:quality} shows the model counts computed by
{\SparseScalMC}, and the bounds obtained by scaling the exact counts
with the tolerance factor $(\varepsilon = 0.8)$ for a small subset of
benchmarks. The $y$-axis represents model counts on log-scale while
the $x$-axis represents benchmarks ordered in ascending order of model
counts. We observe that for {\em all} the benchmarks, {\SparseScalMC}
computed counts within the tolerance.  Furthermore, for each instance,
the observed tolerance ($\varepsilon_{obs}$) was calculated as
max($\frac{|\satisfying{F}|}{\mathrm{AprxCount}}-1,\frac{\mathrm{AprxCount}}{|\satisfying{F}|}-1$),
where $\mathrm{AprxCount}$ is the estimate computed by
{\SparseScalMC}. We observe that the arithmetic mean of
$\varepsilon_{obs}$ across all benchmarks is $0.05$ -- far better
than the theoretical guarantee of $0.8$. 

\section{Conclusion}\label{sec:conclusion}
Our investigations were motivated by the runtime performance of {\SAT} solvers on sparse hash functions. As a first step, we observed that several applications of universal hashing including approximate counting are inherently concerned with concentration bounds provided by universal hash functions. This led us to introduce a relaxation of universal hash functions, christened as $(\rho,\qs, k)$-concentrated hash functions. 
The usage of $(\rho,\qs,k)-$concentrated hash functions ensure that dispersion index for the random variable, $\Cnt{F}{m}$ is bounded by the constant $\rho$.
  We use our bounds to construct sparse hash functions, named $\Hrennes{k}$ where each entry of $A[i]$ is chosen with probability $p_i = \mathcal{O} (\frac{\log_2 i}{i})$. Finally, we replace strong 2-universal hash functions with $\Hrennes{lsa}$ (an analytically computed variant of $\Hrennes{k}$) and implement the resulting algorithm demonstrating significant speedup compared to the state-of-the-art in approximate model counters.

We believe that the concentrated hash functions constructed here could have many potential applications in other domains such as discrete integration, streaming, and the like. This work suggests two interesting directions of future research: 
\begin{itemize}
	\item Design of explicit constructions of sparse hash functions belonging to $(\rho, \qs, k)$-concentrated family for all values of $\qs$, ideally for $\qs = 1$. 
	\item Design of hashing-based techniques where the usage of sparse hash functions performs as good as or better than those based on dense XORs for almost all the benchmarks. 
\end{itemize}

\clearpage
\appendix

\begin{center}
  {\Large \bf Appendix}
  \end{center}
\section{Proofs and details from Preliminaries Section}
\begin{proposition}%
	Let $\mathcal{H}(n,m)$ be a 2-universal hash family and let $h \xleftarrow{R} \mathcal{H}(n,m)$, then $\forall S \subseteq \{0,1\}^n$, we have 
	\begin{align*}
	\expect[|\FullCell{S,h,\alpha}|] = \frac{|S|}{2^m}\\
	\sigma^2[|\FullCell{S,h,\alpha}|] \leq  \expect[|\FullCell{S,h,\alpha} |]
	\end{align*}
\end{proposition}
\begin{proof}
	For $y \in \{0,1\}^n$, define the indicator variable $\gamma_{y, \alpha}$ such that $\gamma_{y, \alpha} = 1 $ if $h(y) = \alpha$ and 0 otherwise. Now, 
	
	\begin{align*}
	 \expect[\gamma_{y, \alpha} ] = \prob[h(y) = \alpha] = \frac{1}{2^m} 
	\text{; Thus, } \\ \expect{|\FullCell{S,h,\alpha}|} = \sum_{y	\in S} \expect[\gamma_{y, \alpha}] = \frac{|S|}{2^m}
	\end{align*}
		Note that $\expect[\gamma_{y, \alpha} \cdot \gamma_{z, \alpha} ] = \prob[h(y) = \alpha \wedge  h(z) = \alpha] = \left(\frac{1}{2^m}\right)^2$.\\ Thus, $\sum_{y,z \in S \mid y \neq z} \expect[ \gamma_{y, \alpha} \cdot \gamma_{z, \alpha}] \leq \frac{|S| (|S|-1)}{2} \left(\frac{1}{2^m}\right)^2 \leq (\expect[|\FullCell{S,h,\alpha}|])^2$
	
	Therefore, 
	\begin{align*}
	\sigma^2_{|\FullCell{S,h,\alpha}|} &= \expect{|\FullCell{S,h,\alpha}|} \\+& \sum_{y,z \in S \mid y \neq z} \expect[ \gamma_{y, \alpha} \cdot \gamma_{z, \alpha}] - (\expect[|\FullCell{S,h,\alpha}|])^2 \\ & \leq \expect[|\FullCell{S,h,\alpha} |]
	\end{align*}

\end{proof}

\concmono*

\begin{proof}
The proof follows immediately from the following three simple observations:
\begin{enumerate}
	\item If a property $\Psi(|S|)$ holds for all $S$ such that $|S| \leq k \cdot 2^m$, then the property $\Psi(|S|)$ also holds for all $S$ such that $|S| \leq k' \cdot 2^m$ for  $k' \leq k$ and $k', k \in \mathbb{N}$. 
	\item If a property $\Psi(m)$ holds for each $m \geq \qs$, then $\Psi(m)$ holds for each $m \geq \qs'$ for $\qs' \geq \qs$. 
	\item $\frac{\sigma^2[\Cnt{S}{m}]}{\expect[\Cnt{S}{m}]} \leq \rho$ implies $\frac{\sigma^2[\Cnt{S}{m}]}{\expect[\Cnt{S}{m}]} \leq \rho'$ for $\rho' \geq \rho$. 
\end{enumerate}
\end{proof}

\probbounds*

\begin{proof}
For every $y \in \{0, 1\}^{n}$ and for every $\alpha \in \{0,
1\}^{i}$, define an indicator variable $\gamma_{y, \alpha, i}$ which
is $1$ iff $h^{(i)}(y) = \alpha$.  Let $\Gamma_{\alpha, i} = \sum_{y
  \in \satisfying{F}} \left(\gamma_{y, \alpha, i}\right)$,
~$\mu_{\alpha, i} = \expect\left[\Gamma_{\alpha, i}\right]$ and
$\sigma^2_{\alpha, i} = \sigma^2\left[\Gamma_{\alpha, i}\right]$.
Clearly, $\Gamma_{\alpha, i} =
\SatisfyingHashed{F}{h^{(i)}}{\alpha}$ and $\mu_{\alpha,
  i} = 2^{-i}|\satisfying{F}|$.  Note that $\mu_{\alpha,i}$
is independent of $\alpha$ and equals $\mu_i$, as defined in the
statement of the Lemma.  By definition of concentrated hash functions, we have  $\frac{\sigma^2_{i}}{\mu_{i}} \leq \rho$ for $|\satisfying{F}| \leq \pivot \cdot 2^i$, i.e., for $ i \geq \log_2 (|\satisfying{F}|) - \log_2 (\pivot)$. Hence statements 1 and 2 of the lemma then follow from Chebhyshev inequality and Paley-Zygmund inequality, respectively.
\end{proof}

\begin{definition}\label{def:var-k-delta}~\cite{EGSS14a}
	Let $Y$ be random variable with $\mu = \expect[Y]$. Then $Y$ is strongly$-(\zeta,\eta)-$concentrated if $\Pr[|Y-\mu| \geq \sqrt{\zeta}] \leq \frac{1}{\eta}$.  
\end{definition}

\begin{proposition}\label{prop:var-k-conc}
	If $\mathcal{H}$ is prefix-$(\rho, \qs, k)$-concentrated family, then for every  $0 < \beta < 1$, $\qs \leq m \leq n$,  and for all $|S| \leq 2^m\cdot k$, then the random variable $\Cnt{S}{m}$ is  strongly-$\left ( (\beta \expect[\Cnt{S}{m}])^2, \frac{\beta^2 \expect[\Cnt{S}{m}]}{\rho}\right)$ concentrated.
\end{proposition}
\begin{proof}
	The proof follows by replacing $(\beta \expect[\Cnt{S}{m}])^2$ by $\zeta$ and $ \frac{\beta^2 \expect[\Cnt{S}{m}]}{\rho}$ by $\eta$ in  Proposition~\ref{prop:probBounds} to obtain that $\Cnt{S}{m}$ is strongly-$\left( (\beta \expect[\Cnt{S}{m}])^2, \frac{\beta^2 \expect[\Cnt{S}{m}]}{\rho}\right)$  concentrated. 
\end{proof}

\subsection{Relationship of Concentrated hashing with other hash families}\label{sec:otherfamilies}
In this section,  we relate other useful notions of hashing to $(\rho,k)$-concentrated hashing. 

\begin{definition}
  A family of hash functions $\mathcal{H}(n,m)$ is
  \begin{itemize}
    \item \emph{uniform} if $\forall x \in \{0,1\}^n$, $\alpha \in \{0,1\}^m, h \xleftarrow{R} \mathcal{H}$, we have %
	$\prob[h(x) = \alpha] = \frac{1}{2^m}$.
    \item \emph{$\varepsilon$- almost universal ($\varepsilon$-AU)} if $\forall x,y \in \{0,1\}^n$ and $\alpha \in \{0,1\}^m$, we have
	\begin{align}
	\prob[h(x) = h(y)] \leq \varepsilon
	\end{align}
        \end{itemize}
\end{definition}

Further, it is known that uniform and $\varepsilon-$AU hash functions allow us to obtain the following concentration bounds.

\begin{proposition}\label{lm:almost-universal-bound}
	Let $\mathcal{H}(n,m)$ be a uniform and $\varepsilon$-almost universal $(\varepsilon$-AU) hash family and let $h \xleftarrow{R} \mathcal{H}(n,m)$, then $\forall S \subseteq \{0,1\}^n$, $|S|\geq 1$, we have 
	\begin{align}
	\expect[|\Cell{S,h,\alpha}|] = \frac{|S|}{2^m}\\
	\sigma^2[|\Cell{S,h,\alpha}|] \leq  \expect[|\Cell{S,h,\alpha}|] +  \frac{(\varepsilon -1)|S|| (|S|-1)}{2^m}
	\end{align}
\end{proposition}
\begin{proof}
	Similar to the above proof, we work with indicator variables $\gamma_{y, \alpha}$ such that $\gamma_{y, \alpha} = 1 $ if $h(y) = \alpha$ and 0 otherwise. Since  $\mathcal{H}(n,m)$ be a uniform, we have $\expect[\gamma_{y, \alpha} ] = \prob[h(y) = \alpha] = \frac{1}{2^m}$. Furthermore,   $\mathcal{H}(n,m)$ is also $(\varepsilon$-AU) , we have $\expect[\gamma_{y, \alpha} \cdot \gamma_{z, \alpha} ] \leq \varepsilon$. Now, substituting the $\expect[\gamma_{y, \alpha} ]$ and $\expect[\gamma_{y, \alpha} \cdot \gamma_{z, \alpha} ]$, we derive the bounds for $\expect[|\FullCell{S,h,\alpha}|]$ and $\sigma^2[|\FullCell{S,h,\alpha}|]$
\end{proof}

Several classical results such as Valiant-Vazirani lemma~\cite{VV85} are typically concerned with upper bounding $\mathcal{G}(|\Cell{S,h,\alpha}|)$ defined as: $\mathcal{G}(|\Cell{S,h,\alpha}|)= \sigma^2[|\Cell{S,h,\alpha}|] - \expect[|\Cell{S,h,\alpha}|] + (\expect[|\Cell{S,h,\alpha}|])^2$. This can indeed be achieved by upper bounding variance using Proposition~\ref{lm:almost-universal-bound}.

It turns out that we can get similar properties with concentrated hash families. Formally,

\begin{proposition}
	If  $\mathcal{H}(n,n)$ is prefix-$(\rho,\qs, k)$-concentrated hash family, then for each $\qs \leq m \leq n$,  $\forall S \subseteq \{0,1\}^n$ where $|S| \leq 2^m\cdot k$, $h \xleftarrow{R} \mathcal{H}$, $\alpha \in \{0,1\}^n$, we have 
	\begin{align}
	\mathcal{G}(\Cnt{S}{m}) \leq  (\rho-1) 	\expect[\Cnt{S}{m}] + (	\expect[\Cnt{S}{m}])^2
	\end{align}
	
\end{proposition}
\begin{proof}
	The proof follows from substituting $\sigma^2[\Cnt{S}{m}] \leq \rho \cdot \expect[\Cnt{S}{m} ]$ in the expression for 	$\mathcal{G}(\Cnt{S}{m}) $
\end{proof}
Just as we replaced 2-universal hash functions with concentrated hash functions for model counting, the above bounds lead us to believe that we can exploit them to replace uniform and $\epsilon$-AU functions by concentrated hash functions in other applications domains such as databases, cryptography and the like. We leave further exploration of this exciting idea for future work.

\section{Proofs from Section~\ref{sec:concentrated}}%
\label{sec:app-conc}
\downleftset*
\begin{proof}
	Similar to~\cite{R18}, the proof strategy is to employ well-known operators whose fixed points reach down-sets and left-compressed sets and prove monotonicity of $\sum_{w=0}^n c_{S}(w)t(w)$ with application of these operators. In what follows, we say that two vectors $x, y\in \{0,1\}^n$ are $i^{th}$-neighbors, denoted $(x,y)\in\nbr_i$, if they differ in coordinate $i$ and are the same elsewhere. 
	
	We first begin with down-set and define, for every $i \in [n]$, an operator $D_i$ on sets $S \subseteq \{0,1\}^n$. The set $D_i (S)$ is obtained from $S$ as follows: Every $z \in S$ is mapped to $\hat{z}$ where 
	\begin{enumerate}
		\item $\hat{z}$ is $i$-th neighbor of $z$ if both $z_i = 1$ and $i$-th neighbor of $z$ is not in $S$. 
		\item $\hat{z} = z$ if $i$-th neighbor of $z$ is in $S$ or $z_i = 0$
	\end{enumerate}
	For example, let $S=\{100,011,101\}$. Then we have $D_3(S)=\{100,$ $010,101\}$ and $D_2(D_3(S))=\{100,000,101\}$. Finally, we get \\$D_1(D_2(D_3(S)))=\{100,000,001\}$, which is a down-set. In fact, it is well-known that for any set $S$, we always have $D(S) := D_1 (D_{2}(\cdots D_n (S)))$ is a down-set. Further, applying the down-operator cannot decrease the expression of interest. An example illustrating this is presented in~\cite{arxiv-report}.
	Formally we have,
	\begin{claim}\label{clm:downset} $\forall i\in [n]$, $\sum_{w=0}^n c_{D_i(S)}(w)t(w) \geq \sum_{w=0}^n c_{S}(w)t(w)$.
	\end{claim}
	\begin{proof} Let us fix $i\in [n]$ and for any $x\in \{0,1\}^{n-1}$, let $x^a$ for $a\in \{0,1\}$ denote the $n$-dimensional vector obtained by inserting $a$ at $i^{th}$ position in $x$. Also, $\mathbb{1}_S(x^a)$ denotes the indicator function, which is $1$ if $x^a\in S$ and $0$ otherwise. Then, %
		\begin{align*}
		\sum_{w=0}^n c_{S}(w)t(w)
		&= \sum_{u,v \in \{0,1\}^n } \mathbb{1}_S(u)\mathbb{1}_S(v) t(d(u, v))\\ &= \sum_{x, y \in \{0,1\}^{n-1}} J_{S}(x, y) \\
		\text{where } J_{S}(x, y)&= \sum_{a,b \in \{0,1\}} \mathbb{1}_S(x^a)\mathbb{1}_S(y^b)t(d(x^a, y^b)) 
		\end{align*}
		Our goal is to compare $\sum_{w=0}^n c_{S}(w)t(w)$ and $\sum_{w=0}^n c_{D_i(S)}(w)t(w)$ by comparing $J_{S}(x, y)$ with $J_{D_i(S)}(x, y)$. Towards this, consider $T=\{x^0,x^1,y^0,y^1\}$. If $S \cap \{x^1, y^1\} = \emptyset$, then $S\cap T = D_i(S)\cap T$, and  $J_{S}(x, y) = J_{D_i(S)}(x, y)$. Therefore, the remaining cases are when there exist $a, b \in \{0,1\}$ such that $\mathbb{1}_S(x^a)\mathbb{1}_S(y^b) = 1$ and $S \cap \{x^1, y^1\} \neq \emptyset$. We then have the following subcases:
		\begin{enumerate}
			\item $x^0 \in S, y ^0 \in S$. In this case $\hat{x^a} = x^a$ and $\hat{y^a} = y^a$ for $a\in\{0,1\}$, which implies $J_{S}(x, y) = J_{D_i(S)}(x, y) $.
			\item $x^0 \in S, y ^0 \notin S, x^1 \in S, y^1 \in S$. Now $\hat{x^1} = x^1,\hat{x^0}=x^0$ and $\hat{y^1}  = y^0$. Since $d(x^0, y^1) = d(x^1, y^0)$, we again have  $J_{S}(x, y) = J_{D_i(S)}(x, y)$ (intuitively, the count $d(x^0,y^1)$ lost because of removing $y^1$ from $S$ in $D_i(S)$ is exactly compensated by $d(x^1,y^0)$ due to adding $y^0$ in $D_i(S)$.)
			\item $x^0 \in S, y ^0 \notin S, x^1 \notin S, y^1 \in S$. Now $\hat{y^1}  = y^0$ and we have $d(x^0, y^1) > d(x^0, y^0)$. Therefore,  $J_{S}(x, y) \leq  J_{D_i(S)}(x, y) $, since $t(w)$ is monotonically non-increasing.
			\item $x^0 \notin S, y ^0 \in S$. The two possibilities arising from this case are symmetric to the above two cases. 
			\item $x^0 \notin S, y ^0 \notin S$. In this case we must have $x^1 \in S$ and $y^1 \in S$ since we know that there exists $a, b \in \{0,1\}$, $\mathbb{1}_S(x^a)\mathbb{1}_S(y^b) = 1$. Thus, we have  $\hat{x^1} = x^0$ and $\hat{y^1} = y^0$. Since $d(x^1, y^1) = d (\hat{x^1}, \hat{y^1})$, we have $J_{S}(x, y) = J_{D_i(S)}(x, y)$.
		\end{enumerate}
		Therefore,  $J_{S}(x, y) \leq J_{D_i(S)}(x, y)$. As this is true for all $x, y\in\{0,1\}^{n-1}$, we conclude that $\sum_{w=0}^n c_{D_i(S)}(w)t(w) \geq \sum_{w=0}^n c_{S}(w)t(w)$ holds for all $i$. 
	\end{proof}
	
	Now moving to the left-compressed set, and we use the operator $L_{i,j}$ on sets $S \subseteq \{0,1\}^n$ for coordinates $i < j \in [n]$. For $z \in \{0,1\}^n$, let $\mathsf{swap}_{i,j}(z)$ represents the vector that is same as $z$ except with the coordinates $i$ and $j$ swapped. The set $L_{i,j}(S)$ is obtained from $S$ as follows: Every $z \in S$ is mapped to $\tilde{z}$ where
	\begin{enumerate}
		\item  $\tilde{z} = \mathsf{swap}_{i,j}(z)$, if $z_i = 0$, $z_j = 1$ and $\mathsf{swap}_{i,j}(z) \notin S$ 
		\item $\tilde{z} = z$, otherwise. 
	\end{enumerate}
	As an example, if we again considering $S=\{100,011,101\}$, then we have $L_{1,2}(S)=S$, $L_{2,3}(S)=\{100,011,110\}$ and $L_{1,2}(L_{2,3}(S))=\{100,101,110\}$ which is a left-compressed set.
	
	We will be interested in the set 
	\begin{align}
	L(S) := L_{1,2}(L_{1,3} (\cdots L_{n-1,n}(S)))
	\end{align}
	and it is easy to see that it is left-compressed. 

	We prove two claims regarding application of $L_{i,j}$ for any $i<j\in [n]$. We fix $i<j\in[n]$ for what follows. For $x\in \{0,1\}^{n-2}$, we let  $x^{ab}$ denote the word $w \in \{0,1\}^n$ such that (i) the $i^{th}$ letter of $w$, $w_i = a$, (ii) the $j^{th}$ letter $w_j = b$ and (iii) removing these two letters in $w$ gives $x$.
	The first property we show is that applying $L_{i,j}$ retains the property of being a down-set. For instance, for the down-set $D(S)=\{100,000,001\}$, $L(D(S))=L_{2,3}(D(S))=\{100,000,010\}$ is also a downset. Formally,
	\begin{restatable}{claim}{downset}\label{claim:downset}
		For down-set $S$, $L_{i,j}(S)$ is also a down-set.
	\end{restatable}
	\begin{proof}
	Fix any $i<j\in [n]$ and consider $x \in L_{i,j}(S)$ and any $y \subseteq x$. If $x \in S$ and $\tilde{x} = x$, then the down-set property of $S$ implies $L_{i,j}(y) =  y$. Assume $x \notin S$, so that $x = w^{10} \in L_{i,j}(S)$ for some $w\in \{0,1\}^{n-2}$ and $x=w^{01} \in S$. There are two possibilities for $y$ to have $y \subseteq x$: either $y =v^{00}$ or $y =v^{10}$ for some $v$. When $y = v^{00}$, then by the down-set property of $S$, we have that $v^{00} \in S$, and thus $v^{00} \in L_{1,2}(S)$. When $y = v^{10}$, then we know $v^{01} \in S$ since $w^{01} \in S$. Therefore, either $v^{10} \in S$ already, or we have $v^{10} \notin S$, which implies $v^{10} \in L_{1,2}(S)$ as desired.
\end{proof}

	We now show the second property, which states that applying the left-compression operator can only increase the sum of interest.
	\begin{restatable}{claim}{leftcomp}\label{claim:leftcomp}
		$\sum_{w=0}^n c_{L_{1,2}(S)}(w)t(w) \geq \sum_{w=0}^n c_{S}(w)t(w)$.
	\end{restatable}
\begin{proof} As before, we start by rewriting,
	\begin{align*}
	\sum_{w=0}^n c_{S}(w)t(w) &= \sum_{x, y \in \{0,1\}^{n-2}} J'_{S}(x, y) \\
	\text{where } J'_{S}(x, y)&= \sum_{a,b,c,d \in \{0,1\} } \mathbb{1}_S(x^{ab}) \mathbb{1}_S(y^{cd}) t\left(d(x^{ab}, y^{cd})\right) 
	\end{align*}
	Let $x,y\in\{0,1\}^{n-2}$. If $x^{aa} \in S$ (resp. $y^{aa}\in S$), then $\widetilde{x^{aa}} = x^{aa}$ (resp. $\widetilde{y^{aa}} = y^{aa}$). %
	Therefore, we need to only consider the expressions and cases depending only on whether  $x^{ab}, y ^{cd} \in S$ or not for $a \neq b$ and $c \neq d$. Again when $S \cap \{ x^{10}, x^{01}, y^{10}, y^{01} \} = \emptyset$, we have  $J'_{S}(x, y) =  J'_{L_{i,j}(S)}(x, y)$. Therefore, for rest of the analysis, we handle the case when $S \cap \{ x^{10}, x^{01}, y^{10}, y^{01} \} \neq \emptyset$.  Let $T' = \{ x^{00}, x^{01}, x^{10}, x^{11}, y^{00},$ $y^{01}, y^{10}, y^{11}  \}$. 
	There are 4 cases:
	\begin{enumerate}
		\item $x^{01}, y^{01} \in S$, then $T' \cap S = T' \cap L_{i,j}(S)$. Therefore,  $J'_{S}(x, y) =  J'_{L_{i,j}(S)}(x, y)$.
		\item $x^{01} \notin S,  y^{01} \notin S$. This can be further subdivided in 4 subcases:
		\begin{itemize}
			\item $x^{10} \in S, y^{10} \in S$. Then $\widetilde{x^{10}} = x^{01}$ and $\widetilde{y^{10}} = y^{01}$. Since $d(x^{10}, y^{10}) = d(x^{01}, y^{01})$, we conclude that \\ $J'_{S}(x, y) =  J'_{L_{i,j}(S)}(x, y)$.
			\item $x^{10} \in S, y^{10} \notin S$. Then, $\widetilde{x^{10}} = x^{01}$. Now notice that for $a,c,d \in \{0,1\}$ we have $d(z^{aa}, x^{cd}) = d(z^{aa}, x^{dc})$. Therefore,  $J'_{S}(x, y) =  J'_{L_{i,j}(S)}(x, y)$.
			\item  $x^{10} \notin S, y^{10} \in S$. This is symmetric to the above case. 
			\item $x^{10} \notin S, y^{10} \notin S$. In this case $S \cap \{ x^{10}, x^{01}, y^{10}, y^{01} \} =  \emptyset$, which is handled above. 
		\end{itemize}
		\item $x^{01} \notin S,  y^{01} \in S$. Again this is subdivided into cases.
		\begin{itemize}
			\item $x^{10} \in S, y^{10} \in S$. Then $\widetilde{x^{10}} = x^{01}$ and $\widetilde{y^{10}} = y^{10}$. Since  $d(x^{10}, y^{10}) = d(x^{01}, y^{01})$ and $d(x^{10}, y^{01}) = d(x^{01}, y^{10})$. Therefore, $J'_{S}(x, y) =  J'_{L_{i,j}(S)}(x, y)$.
			\item $x^{10} \in S, y^{10} \notin S$.  Then $\widetilde{x^{10}} = x^{01}$. Since $d(x^{10}, y^{01}) > d(x^{01}, y^{01})$, we have $J'_{S}(x, y) \leq  J'_{L_{i,j}(S)}(x, y)$ since $t(w)$ is monotonically non-increasing.
			\item $x^{10} \notin S, y^{10} \in S$. Then $\widetilde{y^{10}} = y^{10}$.  Then,  $T' \cap S = T' \cap L_{i,j}(S)$. Therefore, $J'_{S}(x, y) =  J'_{L_{i,j}(S)}(x, y)$.
			\item $x^{10} \notin S, y^{10} \notin S$. Again,  $T' \cap S = T' \cap L_{i,j}(S)$. Therefore, $J'_{S}(x, y) =  J'_{L_{i,j}(S)}(x, y)$.
		\end{itemize}
		\item $x^{01} \notin S,  y^{01} \in S$. This case is symmetric to the above case. 
	\end{enumerate}
	Therefore, for all the cases, it holds $J'_{S}(x, y) \leq  J'_{L_{i,j}(S)}(x, y)$ for all $x, y \in \{0,1\}^{n-2}$. \end{proof}

	The proofs of both these claims are given in~\cite{arxiv-report}.
	Now, combining the above three claims, we obtain the proof of Lemma~\ref{sec:down-left-set}, since each application of the down-set and left-compression operators can only increase the sum $\sum_{w=0}^n c_S(w)t(w)$. So when we reach a fixed-point wrt both these operators, we are sure that the resulting left-compressed down-set maximizes this sum.
\end{proof}

\section{Proof from Section~\ref{sec:concentrated}}
\rbound*
\begin{proof}
	Observe that 
	\begin{align*}
	\text{for $m>2$}, ln(2m)<&\log_2(2m)<2\log_2(m)\\
	\text{then, }  w \geq \frac{m H^{-1}(\delta)}{16}, m>2 &\implies w \geq \frac{H^{-1}(\delta) \cdot m\cdot \log_2(2m) }{16\cdot 2 \cdot \log_2 m} \\
	&\implies 2p_{m}w \geq \log_2(2m) \geq ln (2m) \\
	&\implies  m\cdot exp(-2p_{m}w) < 0.5
	\end{align*}
	Now, since $(1+x)\leq e^{x}$ for all $x$, we have $r(w,m) \leq \frac{((1+exp(-2p_{m}w))^m-1}{2^m}$. Then, $m\cdot exp(-2p_{m}w) < 0.5$ and $exp(-2p_{m}w)<1$ implies that $(1+exp(-2p_{m}w))^m \leq 1+2m\cdot exp(-2p_{m}w)$. Thus, we have 
	\begin{align*}
	r(w,m) \leq 2^{-m} 2m\cdot exp(-2p_{m}w)\\
	\implies \log_2 r(w,m) \leq -m + 1+ \log_2 (m) - 2p_{m}w\\
	\text{But, we have } 2p_{m}w\geq 2 \left(\frac{16}{H^{-1}(\delta)}\frac{\log_2 m}{m}\right)\left(\frac{m H^{-1}(\delta)}{16}\right) = 2\log_2 m\\
	\implies \log_2 r(w,m) \leq -m + 1+ \log_2 (m) - 2 \log_2(m)\\
	\implies \log_2 r(w,m) \leq -m + 1- \log_2 (m)
	\end{align*}
\end{proof}

\section{Proofs from Section~\ref{sec:sparsescalmc}}
\label{sec:app-sparsescalmc}
\upperbound*
\begin{proof}
		We now wish to simplify the upper bound of $\prob\left[\bad\right]$
	        obtained in Equation~\ref{eq:bad-ub}, i.e., \begin{align}
  \prob\left[\bad\right] \leq 
	\prob\left[\bigcup_{i \in \{1, \ldots n\}} \left(\overline{T_{i-1}}
	\cap T_i \cap (L_i \cup U_i)\right)\right]\end{align}

                We
	make three observations, labeled O1, O2 and O3 below, which follow
	from the
	definitions of $m^*$, $\hiThresh$ and $\mu_i$, and from the monotonicity of 
	$\Cnt{F}{i}$.
	\begin{enumerate}
		
		\item[O1:] $\forall i \leq m^*-3$, it is guaranteed that $\frac{|\satisfying{F}|}{2^i(1+\epsilon)}\geq \thresh$.
     From this it follows that (a) $T_i \cap U_i = \emptyset$ and (b) $T_i\cap L_i = T_i$. 
	Therefore, 
	\begin{align*}
		\bigcup_{i \in \{1, \ldots m^*-3\}} \left(\overline{T_{i-1}}
	\cap T_i \cap (L_i \cup U_i)\right) & \subseteq \bigcup_{i \in \{1, \ldots m^*-3\}} \left(\overline{T_{i-1}}
	\cap T_i \right)\\ & \subseteq \bigcup_{i \in \{1, \ldots m^*-3\}} T_i   \subseteq  T_{m^*-3}
		\end{align*}
	where the last containment follows from Equation~\ref{eq:thresh_monotone} . Hence, 	$\prob\left[\bigcup_{i \in \{1, \ldots m^*-3\}} \left(\overline{T_{i-1}}	\cap T_i \cap (L_i \cup U_i)\right)\right] \leq \prob[T_{m^*-3}]$.  

		\item[O2:] For $i \in \{m^*-2, m^*-1\}$, it similarly follows that $\hiThresh \leq \frac{|\satisfying{F}|}{2^i}(1+\frac{\varepsilon}{1+\varepsilon})$ , we have
$T_i \cap U_i = \emptyset$. Since, $T_i \cap L_i \subseteq L_i$, we have 	$\prob\left[\bigcup_{i \in \{m^*-2, m^*-1\}} \left(\overline{T_{i-1}}
\cap T_i \cap (L_i \cup U_i)\right)\right] \leq \prob[L_{m^*-2}] + \prob[L_{m^*-1}] $. 

		\item[O3:] For $i\geq m^*$, it can be shown in the same vein that $\thresh \geq \frac{|\satisfying{F}|}{2^i}(1+\frac{\varepsilon}{1+\varepsilon})$, which implies that $\overline{T_i}\subseteq U_i$. Now, from Equation~\ref{eq:thresh_monotone}, it follows that for all $j$, $\overline{T_j}\subseteq \overline{T_{j-1}}$.
                  This implies that $\prob[\bigcup_{i \in \{m^*,\ldots |S|\}}
		\overline{T_{i-1}} \cap T_i \cap (L_i \cup U_i)]\leq
		\prob[\overline{T_{m^*}} \cup (\overline{T_{m^*-1}}\cap T_{m^*}\cap (L_{m^*} \cup U_{m^*}))] \leq \prob[\overline{T_{m^*}}\cup L_{m^*} \cup U_{m^*}]\leq \prob[L_{m^*}\cup U_{m^*}]$ 
	\end{enumerate}
	Using O1, O2 and O3, we get $\prob[\bad] \leq \prob[T_{m^*-3}]
	+ \prob[L_{m^*-2}] + \prob[L_{m^*-1}] + \prob[L_{m^*} \cup U_{m^*}]$.
\end{proof}

\section{An illustrative example for Claim~\ref{clm:downset}} 
\label{sec:running-example}
Let $n = 3$, $S=\{001,010,100,101\}$ and let $i=3$ . Then $\sum_{w=0}^n c_s(w)t(w)$ can be expressed as sum of the following 16 non-zero terms as follows (after removing the terms where $\mathbb{1}_S(u,v) = 0$ )
\begin{align*}
\sum_{w=0}^{n=3} c_s(w)t(w) &=  t(d(001,001)) + t(d(001,010)) +  t(d(001,100))  \\ & + t(d(001,101))  + t(d(010,001)) + t(d(010,010))  \\ & + t(d(010,100)) + t(d(010,101))  + t(d(100,001)) \\ & + t(d(100,010)) +  t(d(100,100)) + t(d(100,101)) \\& + t(d(101,001)) + t(d(101,010)) +  t(d(101,100)) \\ &+ t(d(101,101))						
\end{align*}

Note for $x,y \in \{0,1\}^2$, 
Observe that, for $x = 00$ and $y=10$, we have $J_{S} (x, y) = J_{S}(00, 10)= t(d(001,100)) + t(d(000,101))$ 

Overall, below are all the non-zero terms for $J_{S} (x, y)$ for $x, y \in \{0,1\}^2$.  
\begin{align*}
J_{S}(00, 00) &= t(d(001,001)) \\
J_{S}(00, 01) &= t(d(001,010))  \\
J_{S}(00, 10)  &= t(d(001,100)) + t(d(001,101))  \\[2em]
\
J_{S}(01, 00) &= t(d(010,001)) \\
J_{S}(01, 01) &= t(d(010,010)) \\
J_{S}(01, 10) &= t(d(010,100)) + t(d(010,101)) \\[2em]
\
J_{S}(10, 00) &= t(d(100,001)) + t(d(101,001)) \\
J_{S}(10, 01) &= t(d(100,010)) + t(d(101,010)) \\
J_{S}(10, 10) &= t(d(100,100)) + t(d(101,101)) + t(d(100,101)) + t(d(101,100)) 
\end{align*}
We can now verify that $\sum_{w=0}^{3} c_s(w)t(w) = \sum_{x, y \in \{0,1\}^{2}} J_{S}(x, y)$
\subsection*{Continuing the example: applying $D_3$ operator}
Observe that $D_3(S)=\{000,010,100,101\}$. Then, we have 
\begin{align*}
J_{D_{3}(S)}(00, 00) &= t(d(000,000)) = t(0) =  J_{S}(000, 000) \\
J_{D_{3}(S)}(00, 01) &= t(d(000,010)) = t(1) \geq t(2)=J_{S}(000, 010) \\
J_{D_{3}(S)}(00, 10)  &= t(d(000,100)) + t(d(000,101)) = t(1)+t(2) = J_{S}(000, 100) \\[2em]
J_{D_{3}(S)}(01, 00) &= t(d(010,000)) = t(1) \geq  t(2)=J_{S}(010, 000) \\
J_{D_{3}(S)}(01, 01) &= t(d(010,010))  = t(0) =  J_{S}(010, 010)\\
J_{D_{3}(S)}(01, 10) &= t(d(010,100)) + t(d(010,101)) = t(2) + t(3) =  J_{S}(010, 100) \\[2em]
\
J_{D_{3}(S)}(10, 00) &= t(d(100,001)) + t(d(101,001)) = t(2) + t(1) = J_{S}(100, 000)\\
J_{D_{3}(S)}(10, 01) &= t(d(100,010)) + t(d(101,010)) = t(2) + t(3) = J_{S}(100, 010) \\
J_{D_{3}(S)}(10, 10) &= t(d(100,100)) + t(d(101,101)) + t(d(100,101)) + t(d(101,100)) \nonumber \\&= J_{S}(100, 100) 
\end{align*}
Therefore, summing up the above equations), we have $\sum_{w=0}^{3} c_S(w)t(w) \leq \sum_{w=0}^{3} c_{D_3(S)}(w)t(w)$


\begin{thebibliography}{10}

\bibitem{AHT18}
Dimitris Achlioptas, Zayd Hammoudeh, and Panos Theodoropoulos.
\newblock Fast and flexible probabilistic model counting.
\newblock In {\em International Conference on Theory and Applications of
  Satisfiability Testing}, pages 148--164. Springer, 2018.

\bibitem{AT17}
Dimitris Achlioptas and Panos Theodoropoulos.
\newblock Probabilistic model counting with short xors.
\newblock In {\em International Conference on Theory and Applications of
  Satisfiability Testing}, pages 3--19. Springer, 2017.

\bibitem{arxiv-report}
S.~Akshay and Kuldeep~S. Meel.
\newblock Sparse hashing for scalable approximate model counting: Theory and
  practice.
\newblock In {\em arXiv:???}, 2020.

\bibitem{AroBar09}
S.~Arora and B.~Barak.
\newblock {\em Computational Complexity: A Modern Approach}.
\newblock Cambridge Univ.~Press, 2009.

\bibitem{AD16}
Megasthenis Asteris and Alexandros~G Dimakis.
\newblock Ldpc codes for discrete integration.
\newblock Technical report, Technical report, UT Austin, 2016.

\bibitem{BSSMS19}
Teodora Baluta, Shiqi Shen, Shweta Shinde, Kuldeep~S Meel, and Prateek Saxena.
\newblock Quantitative verification of neural networks and its security
  applications.
\newblock In {\em Proceedings of the 2019 ACM SIGSAC Conference on Computer and
  Communications Security}, pages 1249--1264, 2019.

\bibitem{BR17}
Paul Beame and Cyrus Rashtchian.
\newblock Massively-parallel similarity join, edge-isoperimetry, and distance
  correlations on the hypercube.
\newblock In {\em Proc. of SODA}, pages 289--306. Society for Industrial and
  Applied Mathematics, 2017.

\bibitem{BAM20}
Bhavishya, Durgesh Agarwal, and Kuldeep~S. Meel.
\newblock On the size of xors in approximate model counting.
\newblock In {\em Proceedings of International Conference on Theory and
  Applications of Satisfiability Testing}, 2020.

\bibitem{carter1977universal}
J~Lawrence Carter and Mark~N Wegman.
\newblock Universal classes of hash functions.
\newblock In {\em Proceedings of the ninth annual ACM symposium on Theory of
  computing}, pages 106--112. ACM, 1977.

\bibitem{CFMSV14}
S.~Chakraborty, D.~J. Fremont, K.~S. Meel, S.~A. Seshia, and M.~Y. Vardi.
\newblock Distribution-aware sampling and weighted model counting for {SAT}.
\newblock In {\em Proc. of AAAI}, pages 1722--1730, 2014.

\bibitem{CMMV16}
S.~Chakraborty, K.~S. Meel, R.~Mistry, and M.~Y. Vardi.
\newblock Approximate probabilistic inference via word-level counting.
\newblock In {\em Proc. of AAAI}, 2016.

\bibitem{CMV13b}
S.~Chakraborty, K.~S. Meel, and M.~Y. Vardi.
\newblock A scalable approximate model counter.
\newblock In {\em Proc. of CP}, pages 200--216, 2013.

\bibitem{CMV16}
S.~Chakraborty, K.~S. Meel, and M.~Y. Vardi.
\newblock Algorithmic improvements in approximate counting for probabilistic
  inference: From linear to logarithmic {SAT} calls.
\newblock In {\em Proc. of IJCAI}, 2016.

\bibitem{CM05}
Graham Cormode and Shan Muthukrishnan.
\newblock An improved data stream summary: the count-min sketch and its
  applications.
\newblock {\em Journal of Algorithms}, 55(1):58--75, 2005.

\bibitem{DMPV17}
Leonardo Duenas-Osorio, Kuldeep~S Meel, Roger Paredes, and Moshe~Y Vardi.
\newblock Counting-based reliability estimation for power-transmission grids.
\newblock In {\em Proc. of AAAI}, 2017.

\bibitem{EGSS14a}
S.~Ermon, C.~P. Gomes, A.~Sabharwal, and B.~Selman.
\newblock Low-density parity constraints for hashing-based discrete
  integration.
\newblock In {\em Proc. of ICML}, pages 271--279, 2014.

\bibitem{EGSS13a}
S.~Ermon, C.P. Gomes, A.~Sabharwal, and B.~Selman.
\newblock Embed and project: Discrete sampling with universal hashing.
\newblock In {\em Proc. of NIPS}, pages 2085--2093, 2013.

\bibitem{EGSS13b}
Stefano Ermon, Carla~P. Gomes, Ashish Sabharwal, and Bart Selman.
\newblock Optimization with parity constraints: From binary codes to discrete
  integration.
\newblock In {\em Proc. of UAI}, 2013.

\bibitem{EGSS13c}
Stefano Ermon, Carla~P. Gomes, Ashish Sabharwal, and Bart Selman.
\newblock Taming the curse of dimensionality: Discrete integration by hashing
  and optimization.
\newblock In {\em Proc. of ICML}, pages 334--342, 2013.

\bibitem{FJ14}
M.~Fredrikson and S.~Jha.
\newblock {Satisfiability Modulo Counting: A New Approach for Analyzing Privacy
  Properties}.
\newblock In {\em Proc. of CSL-LICS}, pages 42:1--42:10, 2014.

\bibitem{GHSS07}
C.~P. Gomes, J.~Hoffmann, A.~Sabharwal, and B.~Selman.
\newblock Short xors for model counting: from theory to practice.
\newblock In {\em Proc. of SAT}, pages 100--106, 2007.

\bibitem{GSS06}
C.~P. Gomes, A.~Sabharwal, and B.~Selman.
\newblock Model counting: A new strategy for obtaining good bounds.
\newblock In {\em Proc. of AAAI}, volume~21, pages 54--61, 2006.

\bibitem{IMMV15}
Alexander Ivrii, Sharad Malik, Kuldeep~S. Meel, and Moshe~Y. Vardi.
\newblock On computing minimal independent support and its applications to
  sampling and counting.
\newblock {\em Constraints}, pages 1--18, 2016.

\bibitem{JVV86}
M.R. Jerrum, L.G. Valiant, and V.V. Vazirani.
\newblock Random generation of combinatorial structures from a uniform
  distribution.
\newblock {\em Theoretical Computer Science}, 43(2-3):169--188, 1986.

\bibitem{LM17}
Jean-Marie Lagniez and Pierre Marquis.
\newblock An improved decision-dnnf compiler.
\newblock In {\em Proceedings of the Twenty-Sixth International Joint
  Conference on Artificial Intelligence, IJCAI}, volume 2017, 2017.

\bibitem{M99}
David~JC MacKay.
\newblock Good error-correcting codes based on very sparse matrices.
\newblock {\em IEEE transactions on Information Theory}, 45(2):399--431, 1999.

\bibitem{MVCFSFIM16}
Kuldeep~S Meel, Moshe Vardi, Supratik Chakraborty, Daniel~J Fremont, Sanjit~A
  Seshia, Dror Fried, Alexander Ivrii, and Sharad Malik.
\newblock Constrained sampling and counting: Universal hashing meets sat
  solving.
\newblock In {\em Proc. of Beyond {NP} Workshop}, 2016.

\bibitem{R18}
Cyrus Rashtchian.
\newblock {\em New Algorithmic Tools for Distributed Similarity Search and Edge
  Estimation}.
\newblock PhD thesis, 2018.

\bibitem{RR19}
Cyrus Rashtchian and William Raynaud.
\newblock Edge isoperimetric inequalities for powers of the hypercube.
\newblock {\em arXiv preprint arXiv:1909.10435}, 2019.

\bibitem{Roth1996}
D.~Roth.
\newblock On the hardness of approximate reasoning.
\newblock {\em Artificial Intelligence}, 82(1):273--302, 1996.

\bibitem{SangBeamKautz2005}
T.~Sang, P.~Beame, and H.~Kautz.
\newblock Performing bayesian inference by weighted model counting.
\newblock In {\em Prof. of AAAI}, pages 475--481, 2005.

\bibitem{SGM20}
Mate Soos, Stephan Gocht, and Kuldeep~S. Meel.
\newblock Accelerating approximate techniques for counting and sampling models
  through refined cnf-xor solving.
\newblock In {\em Proceedings of International Conference on Computer-Aided
  Verification (CAV)}, 7 2020.

\bibitem{SM19}
Mate Soos and Kuldeep~S Meel.
\newblock Bird: Engineering an efficient cnf-xor sat solver and its
  applications to approximate model counting.
\newblock In {\em Proceedings of AAAI Conference on Artificial Intelligence
  (AAAI)(1 2019)}, 2019.

\bibitem{Stockmeyer83}
L.~Stockmeyer.
\newblock The complexity of approximate counting.
\newblock In {\em Proc. of STOC}, pages 118--126, 1983.

\bibitem{Toda89}
S.~Toda.
\newblock On the computational power of {PP} and {(+)P}.
\newblock In {\em Proc. of FOCS}, pages 514--519. IEEE, 1989.

\bibitem{trevisan2002lecture}
L.~Trevisan.
\newblock Lecture notes on computational complexity.
\newblock {\em Notes written in Fall}, 2002.
\newblock
  \url{http://citeseerx.ist.psu.edu/viewdoc/download?doi=10.1.1.71.9877&rep=rep1&type=pdf}.

\bibitem{V12}
Salil~P Vadhan et~al.
\newblock Pseudorandomness.
\newblock {\em Foundations and Trends{\textregistered} in Theoretical Computer
  Science}, 7(1--3):1--336, 2012.

\bibitem{VV85}
Leslie~G Valiant and Vijay~V Vazirani.
\newblock Np is as easy as detecting unique solutions.
\newblock In {\em Proceedings of the seventeenth annual ACM symposium on Theory
  of computing}, pages 458--463. ACM, 1985.

\bibitem{Valiant79}
L.G. Valiant.
\newblock The complexity of enumeration and reliability problems.
\newblock {\em SIAM Journal on Computing}, 8(3):410--421, 1979.

\bibitem{ZCSE16}
S.~Zhao, S.~Chaturapruek, A.~Sabharwal, and S.~Ermon.
\newblock Closing the gap between short and long xors for model counting.
\newblock In {\em Proc. of AAAI}, 2016.

\end{thebibliography}
\end{document}